\documentclass[preprint]{elsarticle}
\usepackage{natbib}
\usepackage{amsmath}               % great math stu
\usepackage{amsfonts}              % for blackboard bold, etc
\usepackage{amssymb}
\usepackage{amsthm}               % better theorem environments
\usepackage{graphicx,subfigure}

\newtheorem{prop}{Proposition}

\date{\today}

\journal{}

\begin{document}
\begin{frontmatter}

\title{The effects of nutrient chemotaxis on bacterial aggregation patterns with non-linear degenerate cross diffusion}

\author[posgrado]{J. Francisco Leyva}
\ead{jfleyva.84@gmail.com}

\author[ciencias]{Carlos M\'{a}laga}
\ead{cmi@ciencias.unam.mx}

\author[iimas]{Ram\'on G. Plaza\corref{cor1}} 
\ead{plaza@mym.iimas.unam.mx}

\cortext[cor1]{Corresponding author. Tel.: +52 (55) 5622-3567. Fax: +52 (55) 5622-3564.}

\address[posgrado]{Posgrado en Ciencias Matem\'{a}ticas, Facultad de Ciencias, Universidad Nacional Aut\'{o}noma de M\'{e}xico, Circuito Exterior s/n, Ciudad Universitaria, C.P. 04510 M\'{e}xico D.F. (Mexico)}

\address[ciencias]{Departamento de F\'{i}sica, Facultad de Ciencias, Universidad Nacional Aut\'{o}noma de M\'{e}xico, Circuito Exterior s/n, Ciudad Universitaria, C.P. 04510 M\'{e}xico D.F. (Mexico)}

\address[iimas]{Departamento de Matem\'{a}ticas y Mec\'{a}nica, Instituto de Investigaciones en Matem\'aticas Aplicadas y en Sistemas, Universidad Nacional Aut\'{o}noma de M\'{e}xico, Apdo. Postal 20-726, C.P. 01000 M\'{e}xico D.F. (Mexico)}

%\maketitle
\begin{abstract}
This paper introduces a reaction-diffusion-chemotaxis model for bacterial aggregation patterns on the surface of thin agar plates. It is based on the non-linear degenerate cross diffusion model proposed by Kawasaki \textit{et al.} \cite{KMMUS} and it includes a suitable nutrient chemotactic term compatible with such type of diffusion. High resolution numerical simulations using Graphic Processing Units (GPUs) of the new model are presented, showing that the chemotactic term enhances the velocity of propagation of the colony envelope for dense-branching morphologies. In addition, the chemotaxis seems to stabilize the formation of branches in the soft-agar, low-nutrient regime. An asymptotic estimation predicts the growth velocity of the colony envelope as a function of both the nutrient concentration and the chemotactic sensitivity. For fixed nutrient concentrations, the growth velocity is an increasing function of the chemotactic sensitivity.
\end{abstract}

\end{frontmatter}
%{\bf Keywords:} Nutrient chemotaxis; Bacillus subtilis; Nonlinear degenerate diffusion; Front propagation.

% \nocite{*}

\section{Introduction}

The spontaneous emergence of complex patterns is ubiquitous in nature. When faced against hostile environmental conditions, for example, bacterial colonies may exhibit very diverse morphological aggregation patterns. Such conditions can be reproduced during \textit{in vitro} experiments %in a Petri dish 
on the surface of thin agar plates by using a very low level of nutrients, a hard surface (high concentration of agar), or both. Some species (like the bacterium \textit{Bacillus subtilis}) are known to exhibit a great variety of branching patterns depending on the nutrient concentration, softness of the medium, and temperature. For example, when inoculated on a nutrient-poor solid agar, the patterns show fractal morphogenesis similar to diffusion-limited aggregation (DLA) (see, e.g., Matsuyama and Matsushita \cite{MatMat1}, and Ben-Jacob \textit{et al.} \cite{BSTCCV2}). When inoculated on semi-solid agar (softer medium) the bacterial colonies exhibit a dense-branching morphology (DBM) with a smooth colony envelope propagating outward (cf. Ohgiwari \textit{et al.} \cite{OMM}). Finally, when both nutrient and agar's softness further increase, the bacteria form simple circular patterns growing homogeneously (cf. Wakita \textit{et al.} \cite{WKNMM}). A detailed account and comparison of the different experimental results can be found in Kawasaki \textit{et al.} \cite{KMMUS} (see also Golding \textit{et al.} \cite{GKCB} and the references therein). These experimental observations are usually summarized in a morphological diagram like the one shown in Figure \ref{morphdiagram}. The horizontal axis measures the medium softness expressed as $1/C_a$, where $C_a$ is the agar concentration, and the vertical axis is the concentration of the nutrient $C_n$.
\begin{figure}[t]
\begin{center}
\includegraphics[scale=0.55]{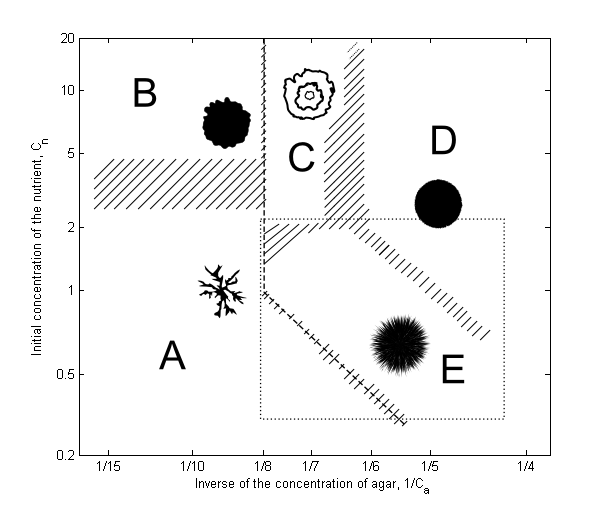}
\end{center}
\caption{Depiction of the morphological diagram of the aggregation patterns exhibited by colonies of \textit{B. subtilis} as a function of the concentration nutrient $C_n$ (vertical axis) and the solidity of the agar medium expressed as $1/C_a$ (horizontal axis). Both concentrations are measured in grams per liter [g/l]. This diagram was redrawn from the one in \cite{KMMUS}.}
\label{morphdiagram}
\end{figure}
When the medium is very hard, the patterns are DLA-like (region A) at low nutrient concentration, and round and compact with fractal boundary for high values of $C_n$ (region B). DBM-like patterns arise in regions where the nutrient is poor and the agar softness is intermediate (semi-solid), with an envelope smoothly rounded and moving outward (region E). In region D, where the agar is softer and the nutrient is more abundant the dense branches fuse together forming homogeneous colonies with smooth boundaries. Although susceptible of further improvements (see, for example, Wu, Zou and Jin \cite{WZJ}), these diagrams provide a good macroscopic description of the colony patterns in terms of simple physical incubation conditions.

In order to explain the experiments, many mathematical models have been proposed. One method of modeling is the continuous deterministic reaction-diffusion approach in which the bacterial density and the nutrient concentration are described with continuous time evolution systems of partial differential equations (cf. Murray \cite{MurII3ed}). Motivated by the experimental observations of Ohgiwara \textit{et al.} \cite{OMM} for DBM-like transition patterns between regions E and D in the morphology diagram (see section \ref{seckawamodel} below), Kawasaki \textit{et al.} \cite{KMMUS} proposed in the mid-90's a reaction-diffusion model with non-linear density-dependent degenerate
(cross) diffusion coefficient, which captures many of the pattern features found experimentally. Thanks to its capability of reproducing the complex dense morphology of patterns as well as its rich mathematical structure, this model has established itself as a well-known non-linear diffusion system in the literature (see, e.g., \cite{MurII3ed,Sh3,SMGA1} and the references therein).

There are, however, other phenomena which must be taken into consideration. Bacterial chemotaxis (that is, the process by which bacteria migrate toward higher concentrations of certain chemical fields) has attracted significant interest due to its critical role in pattern formation (see \cite{KeSe1,KeSe2,Alt1,Ptl1} and \cite{HiPa1} for a review). Ben-Jacob and co-workers (cf. \cite{CCB,GKCB,BeJ2}) have suggested that the chemotactic response of the bacteria is essential to understand some of the fundamental features of the observed patterns. In particular, Golding \textit{et al.} \cite{GKCB} discussed the experimental evidence and provided theoretical arguments favouring the introduction of chemotaxis towards high concentration of nutrients into many of the the reaction-diffusion models existing in the literature. Although the authors did not propose a new chemotactic term for the particular model of Kawasaki and collaborators, they certainly discussed what the relation between a generic diffusion term and the appropriate chemotaxis should be (see section \ref{secrolechemo} below).

Following these ideas, this paper incorporates a suitable nutrient chemotactic term into the original model by Kawasaki \textit{et al.}, and explores the effects that such chemical signaling brings upon the aggregation patterns. For that purpose, we have conducted high resolution numerical simulations of the new system of equations  using Graphic Processing Units (GPUs) to perform  parallel computations. Moreover, we were able to reproduce the results without chemotaxis and to compare them to the new emerging patterns when the chemical signaling is switched on. One of the main features of the new chemotactic term is an increase of the front propagation velocity for the colony envelope. Such observation is justified using detailed front asymptotic calculations of the speed and the numerical results. We find that, except for a regime of soft agar and low nutrient, the chemotaxis does not modify the pattern colony structure significantly, providing evidence of the conjecture of Cohen \textit{et al.} \cite{CCB}.

The remainder of this article is organized as follows. Section \ref{secmathmodel} gives account of the original non-linear reaction-diffusion (non-chemotactic) model and introduces the chemotactic term into the equations. The resulting system is further  normalized in order to reduce the number of parameters. Section \ref{secnumerics} contains the results of our numerical simulations, as well as a comparison of the observed patterns to those without chemotaxis. In order to explain the effects of the chemotactic terms, section \ref{secspeed} contains an approximation of the speed of the envelope front using the ideas of geometric front propagation. Furthermore, we provide a numerical estimation of the speed which approximates well the theoretical result, and compare it to the speed without chemotaxis. The estimations are supported by some theoretical calculations which can be found in \ref{secspeedcomp} at the end of the paper. Finally, in section \ref{secdiscussion}, we make some concluding remarks.

\section{Mathematical modeling}
\label{secmathmodel}

In this section we introduce the mathematical reaction-diffusion-chemotaxis model studied in this paper. We start by providing a detailed explanation of the original model in \cite{KMMUS}. 
% 
% In section \ref{secrolechemo} we justify the choice of the nutrient chemotactic term. The central section \ref{secourmodel} introduces the full reaction-diffusion-chemotaxis model studied in this paper. Finally, we non-dimensionalize the resulting system of equations in order to reduce the number of free parameters.

\subsection{The model of Kawasaki \textit{et al.}}
\label{seckawamodel}

The model proposed by Kawasaki \textit{et al.} \cite{KMMUS} is a reaction-diffusion continuous system of partial differential equations that takes into account the detailed movement of the bacteria. The model has the following general form:
\begin{subequations}\label{ka-model-dim}
\begin{align}
n_t &= D_n \Delta n -f(n,b) \label{eqnkawa}\\
b_t &= \nabla \cdot (D_b \nabla b) + \theta f(n,b) \label{eqbkawa}
\end{align}
\end{subequations}
where $n= n(x,y,t)$ and $b= b(x,y,t)$ represent the concentration of the nutrient and the density of the bacterial cells, respectively. We will denote every point in the domain as $(x,y) \in \Omega \subset \mathbb{R}^2$. Here $\Omega$ is a bounded domain and $t \geq 0$ denotes time. This model assumes that the bacteria and the nutrient diffuse, with diffusion coefficients $D_b$  and $D_n$, respectively. The coefficient $D_n>0$ is constant.

The key feature of this model is the choice of the nonlinear diffusion coefficient $D_b$, which depends on both $b$ and $n$. Its design was motivated by the work of Ohgiwari \textit{et al.} \cite{OMM} who investigated the patterns formed by \textit{B. subtilis} on thin agar plates. They discovered that variations in the environmental conditions as well as in the concentrations of both nutrient and agar medium lead to drastic changes in the morphology of the colony patterns. They found two distinct type of growing processes.  One is performed by immobile cells  when the growth occurs on hard agar plates. In this case, the outermost part of the colonies grows by cell division and cells in the inner part change into a dormant state. The other type of growth takes place for intermediate and low values of $C_a$. This means that the growth is produced by the movement of bacterial cells and, due to the softness of the agar, they can move around dynamically. The outermost front of the colonies is a boundary layer of cells whose movement is dull. But cells inside the layer move around actively and sometimes they break through the layer and immediately become dull. The bacteria inside the colonies are usually inactive. 

In view of such experimental observations, Kawasaki and co-workers conjectured that the bacteria are immotile when either $n$ or $b$ are low and become active as $n$ or $b$ increase. Therefore, they proposed $D_b$ as a non-linear (cross) diffusion function of the form
\begin{equation}\label{CoefCrsDif}
D_b= \sigma nb,
\end{equation}
where $\sigma= \sigma_0 (1+ \Delta)$. The parameter $\sigma_0>0$ is constant and represent the agar concentration (its value increases as the concentration of the agar decreases), whereas the parameter $\Delta$ captures the random diffusion of the cells as a random fluctuation. 

Finally, the term $f(n,b)$ is the consumption rate of the nutrient by the bacteria and $\theta f(n,b)$ is the growth rate of cells, being $\theta>0$ the conversion rate factor. It is assumed that the function $f$ is of Michaelis-Menten type,
\begin{equation}
\label{MicM}
f(n,b)= \frac{\kappa nb}{1+\gamma n},
\end{equation}
where $\kappa >0$ and $\gamma>0$ are considered as constants. In the case when the concentration of the nutrient is low, the function \eqref{MicM} can be approximated by 
%the first non-zero term of its Taylor expasion around $n=0$, that is
\begin{equation}\label{goodkinetics}
f(n,b)= \kappa nb.
\end{equation}
The latter is the form of the kinetic function used by Kawasaki and collaborators in their simulations, and the function that we consider in this paper as well.

\subsection{The role of chemotaxis}
\label{secrolechemo}

Golding \emph{et al.} \cite{GKCB} argue that the growth of bacterial colonies is an auto organization process which involves several mechanisms of chemotactic responses. They suggested that there are at least three different types of chemotactic signals. One of the signals is the result of nutrient (or food) attractive chemotaxis and it is the 
%which, according to the ``receptor's law" (cf. \cite{LapSch1,ForLauf1,Seg1,Seg2}), it will be the 
dominant signal in certain regimes of the nutrient level. Another type of signal is a long range repulsive one, which is produced by the starving bacteria in the center of the colonies. The last type of signal is secreted by the bacteria in the front that are immersed in waste products. These bacteria ask for help to metabolize waste products by segregating a chemoattractant. The range of this signal is determined by the diffusion coefficient and the rate of decomposition. It is, in general, of short range.

When refering to bacterial colonies, Golding \emph{et al.} \cite{GKCB} suggested that nutrient chemotaxis is the mechanism responsible of the increment in the growth speed and the control (or even the decrease) of the fractal dimension of the arising patterns. Such process should permit to increase the propagation of the front.

In view of these observations, we introduce a nutrient chemotactic term into system \eqref{ka-model-dim}. For that purpose, we need to define the chemotactic flux $J_c$ which, in general, is written as
\begin{equation}
J_c =\xi (b) \chi (n) \nabla n.
\end{equation}
Here $\chi = \chi (n)$ is the chemotactic sensitivity function and $\xi = \xi (b)$ is the bacterial response to the nutrient gradient. The chemotactic signals are detected by dedicated membrane receptors which perform a binding between ligands (attractors) and specific binding proteins \cite{WadArm1}. Bacteria sense temporal gradients, that is, they compare sequential registers of their occupied chemoreceptors \cite{Eisen1}. This complexity is built into the model equations through a signal-dependent chemotactic sensitivity function. One of the most commonly used forms is the ``receptor law'' proposed by Lapidus and Schiller \cite{LapSch1},
\begin{equation}\label{receptor_law}
\chi(n) = \frac{\chi_0 K_d}{(K_d+n)^2}
\end{equation}
where $\chi_0>0$ is a constant measuring the strength of the chemotaxis and $K_d > 0$ is the receptor-ligand binding dissociation constant, which has nutrient concentration units, and represents the nutrient level needed for half of receptor to be occupied. This receptor sensitivity law has been derived and applied in several models for chemotaxis (cf. \cite{ForLauf1,Seg1,Seg2}). Its design was motivated by the experimental observation of Mesibov, Ordal and Adler \cite{MOA} that for very high concentrations of nutrient the chemotactic response vanishes due to saturation of the receptors. This feature was not previously captured by the standard Keller-Segel chemotactic flux function (see \cite{LapSch1}). The functional form of \eqref{receptor_law} can be understood in terms of a simple model for receptor-signal binding through the law mass action (see \cite{OtSt1,MOA}). It is to be noted that the dissociation constant $K_d$ for the attractor-receptor molecular interaction has a unique value that depends on the bacterial strain, nutrient type and other \textit{in vitro} conditions, and has to be determined experimentally (see Goldman and Ordal \cite{GolOrd1} for values of $K_d$ on various substrates of \textit{B. Subtilis}). Moreover, it can be interpreted as the nutrient level where the chemotactic sensitivity $\chi(n)$ is maximum. Thus, we shall use $K_d$ as the normalization factor for the nutrient density $n$.

The bacterial response function $\xi = \xi(b)$ is positive for attractive chemotaxis and negative for a repulsive one.
In \cite{GKCB} the authors suggest that, if the bacteria move within a liquid and the colony density is low, then the bacterial response should be proportional to the bacteria concentration times the diffusion, that is,
\begin{equation}
\label{proposalchem}
| \xi (b)| \propto b D_b. 
\end{equation}
Notice that, on account of \eqref{proposalchem}, if the diffusion coefficient is constant then the bacterial response function is proportional to $b$, recovering in this case the classical Keller-Segel chemotactic flux \cite{KeSe1,KeSe2}. Therefore, the authors suggest that, for general non-linear diffusion functions $D_b$, the expression of the chemotactic flux $J_c$ must be modified according to \eqref{proposalchem} in order to incorporate the density dependence of the bacterial movement. For example, for the non-linear diffusion considered by Kitsunezaki \cite{Kitsu}, namely $D_b=D_0 b^m$ where $D_0$ is constant and $m > 0$ is an integer, they proposed $\xi(b)= b D_b= D_0 b^{m+1}$ as the appropriate response function.

\subsection{Incorporation of the nutrient chemotactic term}
\label{secourmodel}

We now introduce a suitable chemotactic term into model \eqref{ka-model-dim} expressing the chemical signaling of the nutrient concentration $n$. In view of the form of the non-linear diffusion \eqref{CoefCrsDif} and taking into account \eqref{proposalchem}, we propose the following (attractive) chemotactic flux
\begin{equation}
\label{chemflux}
J_c = \sigma  n b^2 \chi(n) \nabla n,
\end{equation}
where $\chi(n)$ is the receptor law \eqref{receptor_law}. Substracting the divergence of such flux to the bacterial density equation in \eqref{ka-model-dim} we obtain the following reaction-diffusion-chemotaxis model:
\begin{subequations}
\label{ka-chem-model-dim}
\begin{align}
n_t &= D_n \Delta n -\kappa n b \label{nutriente-dim}\\
b_t &= \nabla \cdot (\sigma n b \nabla b) + \theta \kappa nb - \nabla \cdot \left( \sigma n b^2 \frac{ \chi_0 K_d}{(K_d+n)^2} \nabla n\right), \label{bacteria-dim}
\end{align}
\end{subequations}
for $(x,y) \in \Omega$, $t \geq 0$. The system is subject to initial conditions modeling an uniformly distributed initial constant concentration of nutrient and the initial inoculum of the bacteria. They have the form
\begin{equation}\label{init-cond}
n(x,y,0) = \hat n_0, \quad b(x,y,0)= \hat b_0(x,y).
\end{equation}
Here $\hat n_0$ is a positive constant and $\hat b_0$ is a given function. The system (\ref{ka-chem-model-dim}) is further endowed with no-flux boundary conditions, 
\begin{equation}
\label{boundcond}
\nabla b \cdot \nu = 0, \quad \nabla n \cdot \nu = 0, \qquad (x,y) \in \partial \Omega,
\end{equation}
where $\nu \in \mathbb{R}^2$, $|\nu|=1$, is the outer unit normal vector at the boundary of the domain.

System of equations \eqref{ka-chem-model-dim} adds the nutrient chemotactic signaling to model \eqref{ka-model-dim}. The flux \eqref{chemflux} can be interpreted as the cross chemotactic function corresponding to the nonlinear diffusion \eqref{CoefCrsDif}. %introduced in \cite{KMMUS}.

\subsection{Normalization}
\label{secnormalization}
In order to reduce the number of parameters in system \eqref{ka-chem-model-dim} we introduce the following rescaling of the variables:
\[
 \tilde x = \Big( \frac{\theta K_d \kappa}{D_n} \Big)^{1/2} x, \quad \tilde y = \Big( \frac{\theta K_d \kappa}{D_n} \Big)^{1/2} y, \quad \tilde t = (\theta K_d \kappa) t,
\]
\[
 \tilde n = \frac{n}{K_d}, \quad \tilde b = \frac{b}{\theta K_d}, \quad \tilde \sigma = \Big(\frac{\theta K_d^2}{D_n}\Big) \sigma.
\]
Upon substitution into \eqref{ka-chem-model-dim} and dropping the tilded notation for simplicity we obtain the following dimensionless system of equations
\begin{subequations}\label{ka-chem-model}
\begin{align}
n_t &= \Delta n - n b \label{eqnadim}\\
b_t &= \nabla \cdot (\sigma n b \nabla b) + nb - \chi_0 \nabla \cdot \left( \frac{\sigma n b^2}{(1+n)^2} \nabla n\right), \label{eqbadim}
\end{align}
\end{subequations}
with $\sigma = \sigma_0(1+\Delta)$, and with initial conditions
\begin{equation}
\label{inicond}
\begin{aligned}
n(x,y,0) &= \hat n_0 / K_d \equiv n_0, \\
b(x,y,0) &= \hat b_0(x,y,0) / (\theta K_d) \equiv b_0(x,y,0),
\end{aligned}
\end{equation}
for $(x,y) \in \Omega$. To complete the system we impose no-flux boundary conditions on the rescaled variables
\begin{equation}
\label{bconds}
\nabla b \cdot \nu = 0, \quad \nabla n \cdot \nu = 0, \qquad (x,y) \in \partial \Omega.
\end{equation}

The only remaining free parameters are $\sigma_0 > 0$, measuring the hardness of the agar medium (large $\sigma_0$ for low agar concentration); $n_0 > 0$, measuring the relative initial nutrient concentration; and $\chi_0 \geq 0$, which measures the intensity of the chemotaxis. 

In the sequel we compute numerically the solutions to system \eqref{ka-chem-model} under conditions \eqref{inicond} and \eqref{bconds}. We are particularly interested in the interplay between the different parameter values expressing agar concentration, initial nutrient and the chemotactic signal, and how they affect the colony patterns.

\section{Numerical simulations} 
\label{secnumerics} 

We performed several numerical simulations of the system \eqref{ka-chem-model} using an explicit second order Runge-Kutta and finite differences numerical scheme and the NVIDIA CUDA libraries for the C++ programming language. The libraries are written to run on GPUs, and the CUDA programming language is designed under the simple program/multiple data (SPMD) programming model (cf. \cite{Kirk}). More specifically, the GPU architecture is very well-suited to address problems that can be expressed as data-parallel computations, that is, the same program (finite difference calculations) is executed on many data elements (grid points) simultaneously (see \cite{Nvidia}). Hence, we overcome the inconvenience of small time steps required by explicit numerical schemes for solving stiff equations, and were able to accomplish thousands of iterations in a couple of hours. Our numerical experiments were performed on a NVIDIA Tesla$^\copyright$ C2070 graphics card with 448 CUDA cores.

The calculations were performed on a two dimensional square domain of side $L = 680$. The selected mesh for all our computations is a square grid of $N=2048$ points per side (see discussion on section \ref{secgridsize} below). Therefore, the grid width that we used was $\Delta x= \Delta y= L/N \approx 0.3320$. We considered two values for the time step, namely, $\Delta t= 0.0231$ and $ \Delta t= 0.0110$, for which the method showed to be stable.

Following \cite{KMMUS} and for the sake of comparison, the initial distribution for the bacteria was taken as
\begin{equation}\label{inibacexp}
b_0(x,y,0)= b_M e^{-(x^2 +y^2)/ 6.25}
\end{equation}
where $b_M=0.71$ is the maximum density at the center of the domain. Recall that the initial condition for the nutrient is a fixed constant, for which we take the values $n_0 = 0.71$ and $1.07$.

Similarly as the computations in \cite{KMMUS}, and in order to incorporate the fact that bacteria move according to a biased random walk, on every grid point and at each time step the diffusion coefficient for the bacteria $\sigma$ is perturbed from the mean $\sigma_0$ through a random variable $\Delta$ taken from a triangular distribution with support on $[-1,1]$, that is, $\sigma = \sigma_0 (1+\Delta)$. Remember that $\sigma_0$ represents the softness of the agar medium.

Regarding the chemotactic term, the coefficient $\chi_0$ represents the strength of the signal. The simulations were performed using several values of this parameter including the value $\chi_0 = 0$, that is, in the absence of chemotaxis. Thus, we also computed solutions to the original system \eqref{ka-model-dim} as well (in its non-dimensional form) also for the sake of comparison. For a complete summary of the free parameter values used during the numerical experiments, see table \ref{tabresparam}. 

\begin{table} 
\begin{center}
\begin{tabular}{|l|l|l|}
  \hline
Description & Symbol & Values \\ \hline \hline
  %\multicolumn{3}{|c|}{Free parameter values} \\
  %\hline
  Initial condition for the nutrient & $n_0$ & 1.07, 0.71  \\
  Agar softness &$\sigma_0 $ & 1.0, 4.0  \\
  Chemotactic signal strength &$\chi_0$ & 0, 2.5, 5.0, 7.5, 10.0 \\
  %Conversion rate (Consumed/Growth) & $\theta$ & 1 \\

  \hline

\end{tabular}  
\caption{Free parameter values used during the numerical simulations.}
\label{tabresparam}
\end{center}
\end{table}

\subsection{Results}
\label{secresults}
We now present the numerical simulations of system \eqref{ka-chem-model} subject to initial conditions \eqref{inibacexp} and $n(x,y,0) = n_0$, as well as to boundary conditions \eqref{bconds}. We divide the exposition of our numerical results in terms of the agar's hardness. The latter is expressed by the parameter value $\sigma_0$. In our simulations, we considered the values $\sigma_0 = 1$ and $\sigma_0 = 4$, and the observed patterns correspond to regions E and D in the morphology diagram (see the enclosed region inside a rectangle in figure \ref{morphdiagram}). For $\sigma_0=1$ we took the time step $\Delta t= 0.0231$ and for $\sigma_0=4$ we took it as $ \Delta t= 0.0110$. The computed solutions are represented in figures \ref{fig:set1} through \ref{fig:set4}. Each figure contains two columns associated to two different values of the parameter $\chi_0$, and each row is associated to a different time step. Such configuration allows the reader to compare the colony patterns for different chemotactic intensities. In what follows we highlight that the main effect of the chemotaxis is the increment on the growth speed of the outer envelope of the bacterial colony.

\subsubsection{Semi-solid agar}
We begin with the case of semi-solid agar, for which $\sigma_0 = 1$. This means that the concentration of agar is intermediate and the diffusion of the bacteria is limited. The observed patterns correspond to the DBM-like branching structure in region E of the morphology diagram. The considered values for the initial nutrient level were $n_0= 0.71$ and $1.07$. 

Figure \ref{fig:set1} shows the numerical calculations for $n_0 = 0.71$. On the left column we find the solutions to the system in the absence of chemotaxis, that is, when $\chi_0 = 0$. The pictures for this case are comparable to those obtained by Kawasaki \textit{et al.} (see figure 3(b) in \cite{KMMUS}). On the right column we find the case where the chemotactic sensitivity is set as $\chi_0=5$. The morphology shown in both cases is very similar (if we look carefully we can see that the number of branches is consistent). There is, however, a significant increase in the growth velocity of the outer envelope of the branches when the chemotactic signal is switched on, as the reader may observe.

% \begin{figure}[!ht]
% \begin{center}
% \subfigure[$\chi_0 = 0$, $t = 694.5$]{
%             \label{set1a}
%             \includegraphics[scale=0.125]{s1q071x0t1}
%         }
% \subfigure[$\chi_0= 5$, $t = 694.5$]{
%             \label{set1b}
%             \includegraphics[scale=0.125]{s1q071x5t1}
%         }\\
% \subfigure[$\chi_0 = 0$, $t = 1389.0$]{
%             \label{set1c}
%             \includegraphics[scale=0.125]{s1q071x0t2}
%         }
% \subfigure[$\chi_0= 5$, $t = 1389.0$]{
%             \label{set1d}
%             \includegraphics[scale=0.125]{s1q071x5t2}
%         }\\
% \subfigure[$\chi_0 = 0$, $t = 2083.5$]{
%             \label{set1e}
%             \includegraphics[scale=0.125]{s1q071x0t3}
%         }
% \subfigure[$\chi_0= 5$, $t = 2083.5$]{
%             \label{set1f}
%             \includegraphics[scale=0.125]{s1q071x5t3}
%         }\\
% \subfigure[$\chi_0 = 0$, $t = 2824.0$]{
%             \label{set1g}
%             \includegraphics[scale=0.125]{s1q071x0t4}
%         }
% \subfigure[$\chi_0= 5$, $t = 2824.0$]{
%             \label{set1h}
%             \includegraphics[scale=0.125]{s1q071x5t4}
%         }\\
% \end{center}
% \caption{Colony growth as a result of simulations of system \eqref{ka-chem-model}, taking $\sigma_0 = 1.0, n_0 = 0.71$, with chemotaxis sensitivity $\chi_0 = 0$ (no chemotaxis - left), and $\chi_0 = 5.0$ (right), for different values of $t$.
% %Iterations are 30000, 60000, 90000, 122000 and the times are 694.5, 1389, 2083.5, 2824.
% }
% \label{fig:set1}
% \end{figure}
% 

The simulations for $n_0=1.07$ are depicted in figure \ref{fig:set2}. Once again the left column corresponds to the value $\chi_0 = 0$. The patterns are defined by fine radially oriented branches with round envelope. At the center of the colony there is the highest concentration of bacteria. Actually, the peak concentration is the remainder of the initial inoculation, because the bacteria at early times deplete the nutrient in site of the initial condition, but the cells in the inner part of the colony become immotile. Meanwhile, the cells at the periphery of the colony move actively searching for higher concentrations of nutrient; this is what drives the formation of branches. Hence, bacteria at the front of the colony account for the instability at the interface.  This is the main feature of the non-linear cross diffusion model. In the right column of figure \ref{fig:set2} we find the patterns in the presence of a chemotactic signal ($\chi_0=5.0$). We observe that at early stages (see figure \ref{set2b}) the chemotaxis provides a significant outward drift to the cellular movement as discussed in \cite{GKCB}, causing a more ramified pattern with round envelope and a substantial enhancement of growth speed. Comparing figures \ref{set2g} and \ref{set2h} we notice that the chemotactic signal produced a pattern two times bigger (the envelope is moving faster).

\subsubsection{Soft agar}

In our next numerical experiment, we set the value $\sigma_0 = 4$, causing the diffusion of the bacteria to be increased by lowering the concentration of the agar. At sufficient high levels of nutrient the colony patterns are restricted to region D in the morphology diagram, whereas at lower levels of $n_0$ patterns in region E emerge. In order to explore the morphology transitions between regions D and E, we first set $n_0 = 0.71$ as the initial nutrient and consider various values for $\chi_0$, because the chemotactic signal gives rise to a outward drift to the cellular movement. The computations are depicted in figures \ref{fig:set3a} and \ref{fig:set3b}. In this experiment, the comparison is made through the two figures, for which the initial nutrient is fixed and the chemotactic sensitivity takes the values $\chi_0 =0, 2.5, 5.0$ and $7.5$. In figure \ref{fig:set3a}, the left column shows the evolution of the system without chemotaxis. At early times the bacteria evolve as a solid colony (figure \ref{set3aa}), but for the consequent iterations they develop a pattern of radially oriented fjords. On the right column of figure \ref{fig:set3a} we show the results with chemotactic strength $\chi_0=2.5$. As expected, the chemotactic term makes the colony grow faster; compare, for example, figures \ref{set3af} and \ref{set3ag}. Observe that the pictures in the two columns of figure \ref{fig:set3a} (with chemotactic sensitivities $\chi_0 = 0$ and $\chi_0 = 2.5$) exhibit similar morphologies, being the envelope velocity the only apparent difference between them.

For higher values of $\chi_0$, however, the growth speed is not only greatly increased, but there is also a change in the morphologies as well. We may appreciate, for instance, that for $\chi_0 = 7.5$ the colony has almost doubled its size at time step $t=324.0$ (see figures \ref{set3ag} and \ref{set3bh}). In this case the increment on the growth speed is accompanied by a notable morphological change. Presumably, at higher values of $\chi_0$ the propagation speed of the front is higher than the diffusion of the nutrient preventing the fjords from forming. Therefore, the thin branches start fusing (figure \ref{set3bg}) until the pattern shows no more branches and becomes an homogeneous disk (figure \ref{set3bh}). These observations suggest that, in the transition between regions E and D in the morphology diagram, the supressing effect of the chemotactic term on the onset of instability (which causes the formation of branches; cf. \cite{ArLe}) is more evident than in other regions of the diagram. An increase of the chemotactic intensity prevents the formation of further branches and results in an homogeneization of the colony morphology, accompanied with the expected increase in the velocity of the overall colony.

Finally, in figure \ref{fig:set4} we present the results of our numerical simulations for the value of nutrient $n_0 = 1.07$ and chemotaxis sensitivities $\chi_0 = 0$ and $\chi_0 = 7.5$. In this case no significant morphological changes are observed. In the left column of figure \ref{fig:set4} the pattern shown is an homogeneous disk with a peak of bacterial concentration at the center. In the right column (when $\chi_0 = 7.5$), as expected, it is observed that the chemotactic signal increases the growth speed and the pattern is also a slightly rugged solid disk, but with a small depression at the center. The latter happens because the softness of the agar allows the colony to spread out before the nutrient is depleted. In this case (which correspons to region D in the diagram), we notice the increment in the outer colony velocity with no significant change in morphology.

\subsection{A note on the grid size}
\label{secgridsize}

It has been suggested that numerical simulations might depend on the grid size \cite{MiSaMa}. To circumvent this possibility, we performed simulations with the following grid sizes as well: $N =1024, 2048, 4096$ and $8192$. The results for the parameter values $\sigma_0=1$, $n_0=0.71$ and $\chi_0= 5$ can be observed in figure \ref{fig:TMalla} for each grid size. Notice that in the case of $N= 1024$ points per side, the solutions exhibit a cross shaped pattern with a rhomboid envelop (figure \ref{1k}). When the grid size is increased to $N = 2048$ there is indeed a significant change in morphology, showing a denser branching pattern with a round envelope (figure \ref{2k}). This pattern does not change significantly, however, when the size grid is increased to $N = 4096$ and $N = 8192$ (figures \ref{4k} and \ref{8k}). The reader may notice that figures \ref{2k}, \ref{4k} and \ref{8k} show very similar patterns. We conclude that a resolution of $N = 2048$ lies above the threshold needed for a reliable display of the solutions of the system. Even though our computations were performed with parallelization techniques on a GPU, the simulations with $N = 8192$ involved a large computational time (there are more than 64 million grid points). Therefore, we adopted $N = 2048$ as the resolution for all the simulations of system \eqref{ka-chem-model} in this paper.

\begin{figure}[ht!]
\begin{center}
\subfigure[$N = 1024$]{
            \label{1k}
            \includegraphics[scale=0.15]{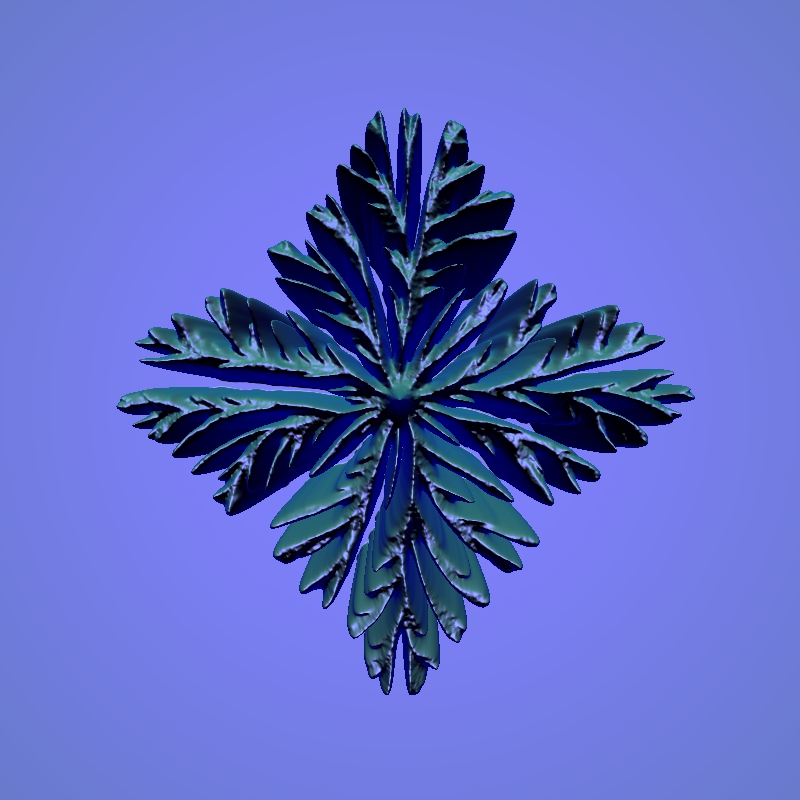}
        }
\subfigure[$N = 2048$]{
            \label{2k}
            \includegraphics[scale=0.15]{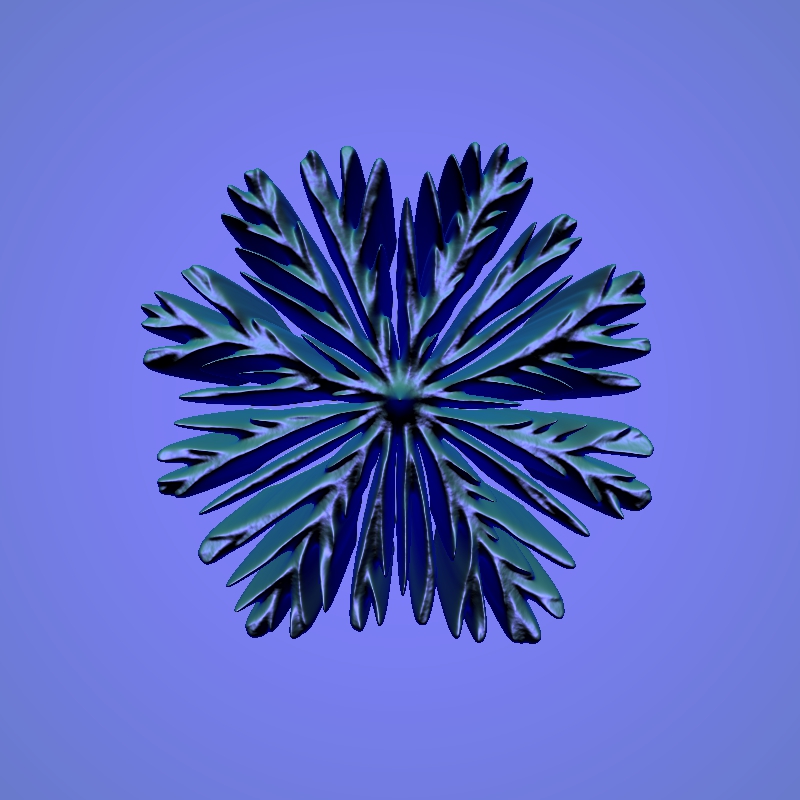}
        }\\
\subfigure[$N = 4096$]{
            \label{4k}
            \includegraphics[scale=0.15]{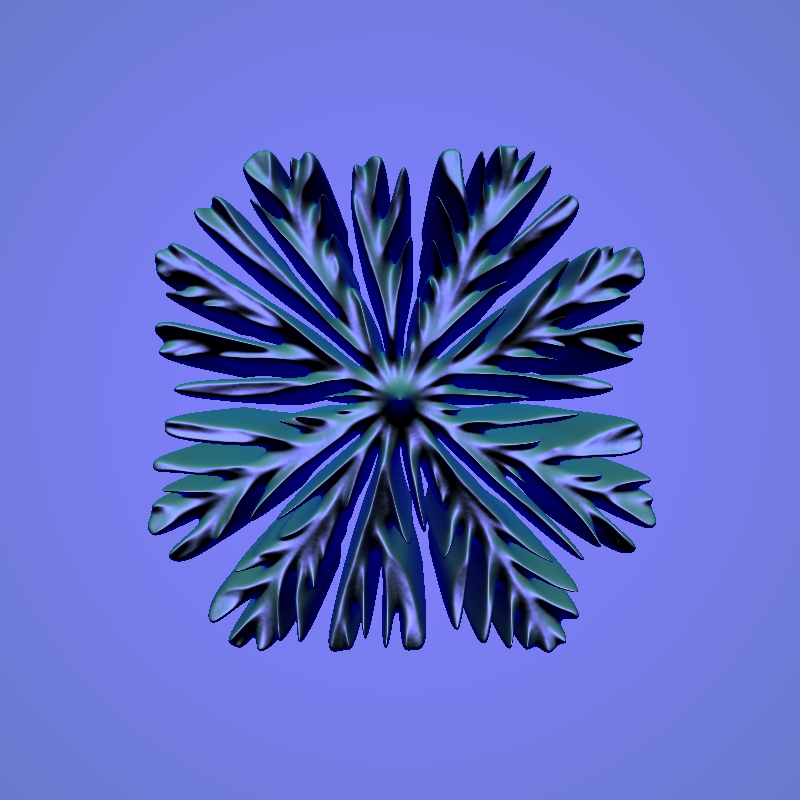}
        }
\subfigure[$N = 8192$]{
            \label{8k}
            \includegraphics[scale=0.15]{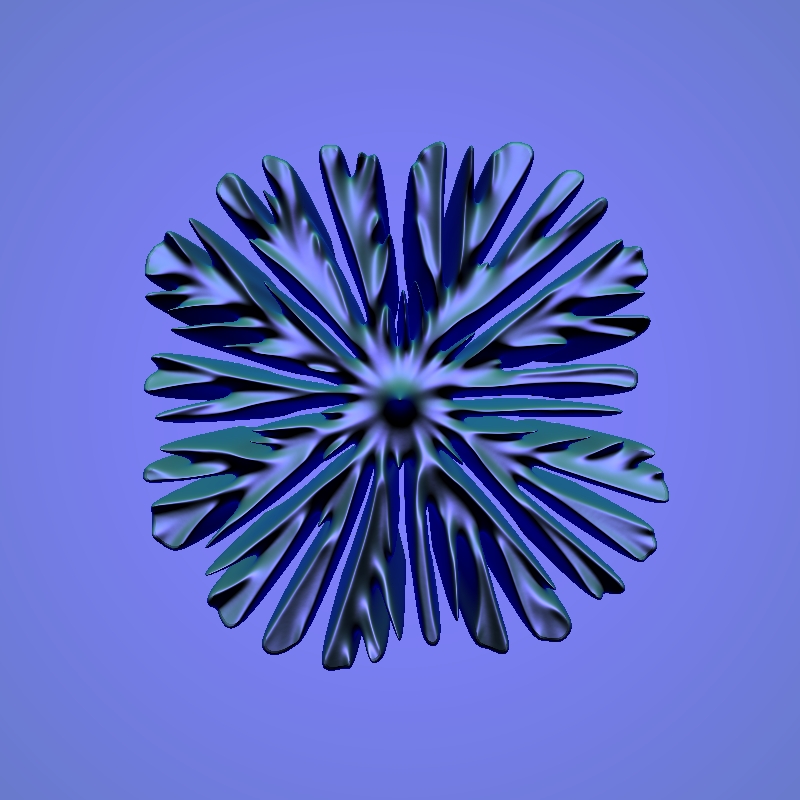}
       }\\
\end{center}
\caption{Simulations of system \eqref{ka-chem-model}, with parameter values $\sigma_0 = 1, n_0 = 0.71$ and $\chi_0=5$, for different values of the grid size $N$.}
%(a) 1024 points per side, (b) 2048 points per side, (c) 4096 points per side, (d) 8192 points per side.}
\label{fig:TMalla}
\end{figure}

\section{Propagation velocity of the front}
\label{secspeed}

In this section we discuss the effects of the chemotactic term on the envelope propagation speed by means of an asymptotic geometric front approximation. Moreover, we provide numerical estimations of the speed by simulating a one dimensional version of system \eqref{ka-chem-model}. We compare these velocity calculations to those without chemotaxis.

\subsection{Approximation of the speed}
\label{secaproxspeed}

For simplicity, let us consider the chemotactic sensitivity function as a constant, i.e., $\chi =\chi_0$. Therefore, the system of equations takes the form
\begin{subequations}\label{velo-system}
\begin{align}
n_t &= \Delta n - n b \label{velo-nut}\\
b_t &= \nabla \cdot (\sigma n b \nabla b) +  nb - \chi_0 \sigma \nabla \cdot \left( n b^2 \nabla n\right), \label{velo-bact}
\end{align}
\end{subequations}
In order to approximate the velocity of the smooth envelope enclosing the branching patterns, we explore the following reduction proposed by Kawasaki \textit{et al.} \cite{KMMUS}. They made a further approximation in order to obtain a scalar equation for the bacterial concentration $b$.
They observed that in the absence of diffusion the total mass is conserved. In our setting, if we ignore both diffusion and chemotaxis then we can add the equations and integrate them in time to obtain
\begin{equation} \label{mass-conser}
n+b=C
\end{equation}
where $C$ is a constant, which will be taken as $C=n_0$. Thus, $n= n_0 -b$ and we obtain the logistic equation
\[
b_t = n_0 b \left( 1- \frac{b}{n_0}\right). 
\]

It must be observed that in \cite{KMMUS} the term $\sigma_0 n b$ is substituted by $\sigma_0 n_0 b$, and the term $nb$ is substituted by $n_0 b (1-b/n_0)$, in the approximation of the system by a scalar equation for $b$. Instead, we are substituting (\ref{mass-conser}) into (\ref{velo-bact}) to obtain a scalar equation for $b$. This not only represents a more accurate approximation but it is also consistent with the numerical computations which, at first order, show that the conservation-like equation \eqref{mass-conser} is valid away from the front.

Substitution of (\ref{mass-conser}) into (\ref{velo-system}) yields the scalar equation 
\begin{equation}\label{scalar-eq-b}
b_t = \nabla \cdot ( \sigma_0 (n_0 - b) b \nabla b)+ (n_0 - b) b - \sigma_0 \chi_0 \nabla \cdot( (n_0- b) b^2 \nabla (n_0-b)),
\end{equation}
which is an approximated scalar equation for $b$ with a Fisher-type reaction term. Here we have approximated $\sigma_0 (1+\Delta) \approx \sigma_0$, because in this part of the study we are not interested in considering the random motion of bacteria. A rearrangement of the terms in equation \eqref{scalar-eq-b} allows us to write the two divergence terms together. This yields
\begin{equation}\label{scalar-eq-b-re}
b_t = \sigma_0 \nabla \cdot (\tilde{D}(b,\chi_0) \nabla b) + g(b),
\end{equation}
where
\begin{equation}\label{D-tilde}
\tilde{D}(b,\chi_0) = n_0 b\left( 1 -\frac{b}{n_0} \right)(1+\chi_0 b),
\end{equation}
may be interpreted as an effective non-linear diffusion coefficient, which involves both the contributions of diffusion and chemotaxis. The reaction term is of Fisher-KPP type,
% 
% In order to track more easily equation \eqref{scalar-eq-b}, we rearranged some terms and incorporated the  contributions  from diffusion and chemotaxis into an effective non-linear diffusion coefficient of the form,
% \begin{equation}\label{D-tilde}
% \tilde{D}(b)= b(n_0-b)(1+\chi_0 b) = n_0 b\left( 1 -\frac{b}{n_0} \right)(1+\chi_0 b).
% \end{equation}
% Then, equation (\ref{scalar-eq-b}) can be written as 
% \begin{equation}\label{scalar-eq-b-re}
% b_t = \sigma_0 \nabla (\tilde{D}(b) \nabla b) + g(b),
% \end{equation}
% where the reaction term is
\begin{equation}\label{reac-term}
g(b)= n_0 b \left( 1 -\frac{b}{n_0} \right).
\end{equation}

We describe the motion of the front in local curvilinear coordinates with components $\zeta(x,y,t)$ and $\tau(x,y,t)$, the normal and tangent unit vectors to the front, respectively. According to custom, $\zeta$ is the normal component pointing inside the colony front, so that $- \zeta_t$ is the outer normal speed. We assume that the dependence of $b$ with respect to the tangent component is negligible and we approximate
\begin{equation}\label{growth-1D}
b(x,y,t)\approx \bar b(\zeta(x,y,t)).
\end{equation}
After substituting $\bar b$ into equation (\ref{scalar-eq-b-re}), and dropping the bar for notation convenience, we get the following ordinary differential equation for $b$
\begin{equation}\label{ode-b}
-c b' = \sigma_0 (\tilde D(b)b')' + b(n_0 -b).
\end{equation}
where
\begin{equation}\label{velocity}
 c = -\zeta_t +\sigma_0 \tilde D(b) \Delta \zeta
\end{equation}
and $-\zeta_t$ is the velocity of the envelope front that we want to estimate.

% Malaguti and Marcelli \cite{MalMar1} have shown that for one-dimensional reaction-diffusion equations of the form
% \begin{equation}\label{scalar-mm}
% u_t = (D(u) u_x)_x + g(u)
% \end{equation}
% with $D \in C^1 ([0,1])$, $g \in C^1 ([0,1]) $, where the diffusion degenerates at $u=0$ and $u=1$, namely, 
% \begin{equation}
% D(0)=D(1)=0, \qquad D(u)>0  \;\;\; \text{for all} \;\; 0 < u < 1,
% \end{equation}
% and the rate of growth of the population is of Fisher-KPP type,
% \begin{equation}
% g(0) = g(1) = 0, \qquad g(u) > 0 \;\;\;\text{for all} \;\; 0 < u < 1,
% \end{equation}
% then there exist a traveling wave solution (unique up to translation) to equation \eqref{scalar-mm}, provided that the velocity $c$ satisfies $0< c_* \leq c $, where $c_*$ is a threshold velocity satisfying the bound 
% \begin{equation}\label{threshold}
% 0 < c_* \leq 2 \sqrt{\sup\limits_{s \in (0,1)} \frac{D(s) g(s)}{s}}.
% \end{equation}
% The case $c=c_*$ corresponds to a sharp front.

The approximated front equation \eqref{ode-b} in the normal direction resembles a one-dimensional equation of form \eqref{theonedimeq} (see \ref{secspeedcomp}). By the results of Malaguti and Marcelli \cite{MalMar1} the one-dimensional velocity for each value of $\chi_0 \geq 0$ can be estimated in terms of $\chi_0$, $n_0$, and the functions $\tilde D$ and $g$. We note that we are interested in comparing the velocity between the cases $\chi_0 = 0$ and $\chi_0 >0$. The aforementioned result shows that one-dimensional traveling wave solutions to equations of the form (\ref{theonedimeq}) must travel with velocity $c(\chi_0)$ bounded by $c(\chi_0) \geq c_*(\chi_0)$, where the threshold velocity $c_*$ satisfies the bound:
\begin{equation}
0 < c_*(\chi_0) \leq \bar{c}(n_0, \chi_0):= 2 \sqrt{\sup\limits_{b \in (0,n_0]} \frac{\sigma_0 \tilde{D}(b,\chi_0) g(b)}{b}}
\end{equation}
with $\tilde{D}$ and $g$ given by (\ref{D-tilde}) and (\ref{reac-term}). Upon substitution, we are able to compute the bound
\begin{equation}
\bar{c} (n_0, \chi_0) = 2 n_0 \sqrt{\sigma_0} \sqrt{\max\limits_{b \in [0,n_0]} b \left(1-\frac{b}{n_0}\right)^2 (1+\chi_0 b)}.
\end{equation} 

The function $\tilde \psi(b) = b (1-b/n_0)^2 (1+\chi_0 b)$ is non-negative for $0 \leq b \leq n_0$ and each $\chi_0 \geq 0$; moreover, it has a global maximum $\tilde \psi_{max}= \tilde \psi(b_*)$, where
\begin{equation}
b_* = \frac{1}{2}\left(\frac{n_0}{2}-\frac{3}{4\chi_0}\right) + \frac{1}{2} \sqrt{\left(\frac{n_0}{2}-\frac{3}{4\chi_0}\right)^2 + \frac{n_0}{\chi_0}},
\end{equation}
and with $0 < b_* < n_0$ for all values of $\chi_0 > 0$ and any fixed value of $n_0 > 0$. Thus,
\begin{equation}\label{c-bar-chi0}
\bar{c}(n_0, \chi_0)= 2 n_0 \sqrt{\sigma_0}\sqrt{\psi(b_*)}.
\end{equation}

In the absence of chemotaxis ($\chi_0=0$) the function $\tilde \psi(b)$ reduces to $\psi(b)= b(1-b/n_0)^2$ and has a global maximum $\psi(n_0/3) = 4 n_0/27$. This means that
\begin{equation}\label{c-bar-0}
\bar{c}(n_0,0)= 2 n_0 \sqrt{\sigma_0}\sqrt{\frac{4 n_0}{27}} = \frac{4}{3\sqrt{3}} n_0^{3/2} \sigma^{1/2}.
\end{equation}
Equations \eqref{c-bar-chi0} and \eqref{c-bar-0} are the speed thresholds for the cases with and without chemotaxis, respectively. These equations predict the speed of a sharp front, and will be used when comparing the numerical estimates of the velocity (see section \ref{secnumspeed} below).
In our setting, we are considering that the growth of the colony is in the radial direction (see equation \eqref{growth-1D}), so we only consider the normal velocity $s = -\zeta_t$. Furthermore, we take into account the cases where the colony exhibits an outer growth, this means that the local curvature is positive there
\begin{equation}\label{curvature}
\kappa = -\Delta \zeta >0.
\end{equation}

Returning to equations (\ref{ode-b}) and (\ref{velocity}) and comparing the one dimensional velocity of a front for equation \eqref{theonedimeq} to expression \eqref{velocity}, we obtain
\begin{equation}
s = -\zeta_t \geq c_*(\chi_0) + \sigma_0 \tilde{D}(b) \kappa >0.
\end{equation} 
It follows from equation (\ref{curvature}) and Proposition \ref{propbound}, in \ref{secspeedcomp}, that 
\begin{equation} 
\label{speedinq}
s= -\zeta_t \geq c_*(\chi_0) + \sigma_0 \tilde{D}(b) \kappa \geq c_*(0) + \sigma_0 \tilde{D}(b) \kappa \geq c_*(0). 
\end{equation}
% 
% \begin{equation} \label{speedinq}
% \begin{split}
% s= -\zeta_t &\geq c_*(\chi_0) + \sigma_0 \tilde{D}(b) \kappa \\
%             &\geq c_*(0) + \sigma_0 \tilde{D}(b) \kappa \\
%             & \geq c_*(0).
% \end{split}
% \end{equation}

Last inequality implies that the normal velocity $s$ is greater than the velocity for the sharp front when there is no chemotaxis ($\chi_0 = 0$), and that it is increased/decreased by a term proportional to the local curvature and the chemotactic strength $\chi_0>0$. In the outer front where diffusion is degenerate, this last curvature term is negligible (as $\tilde D = 0$ in a vicinity of the envelope front) so that we can approximate the speed by the one-dimensional sharp front velocity, $c_*(\chi_0) \lesssim s$, regardless of the curvature sign.

\subsection{Numerical computation of the velocity}
\label{secnumspeed}

As we have stated in (\ref{growth-1D}), the tangential direction is negligible in describing the front movement. Then, we shall estimate the velocity of the envelope front solving numerically a one-dimensional version of system \eqref{velo-system}, that is of the following form,
\begin{subequations}\label{velo-system-1d}
\begin{align}
\frac{\partial n}{\partial t } &= \left( \frac{\partial n}{\partial x } \right)_x  - n b \label{velo-nut-1d}\\
\frac{\partial b}{\partial t } &= \left(\sigma_0 n b \frac{\partial b}{\partial x } - \chi_0 \sigma_0 n b^2 \frac{\partial n}{\partial x } \right)_x +  nb  \label{velo-bact-1d}.
\end{align}
\end{subequations}
The simulations of system \eqref{velo-system-1d} were performed on the spacial domain $[0,20]$, and the initial conditions were take as
\begin{equation}
%\label{init-cond}
n(x,0) = n_0, \; b(x,0)= b_M e^{-{x}^2 / 6.25},
\end{equation}
where $b_M= 0.71$. Again, as in the two-dimensional simulations, $n_0$ measures the initial level of nutrient, $\chi_0$ measures the chemotaxis intensity and $\sigma_0$ measures the agar hardness. We made several simulations of system \eqref{velo-system-1d} for various values of $n_0 = \lbrace 0.35, 0.71, 1.07, 2, 3, 4 \rbrace$, $\chi_0 = \lbrace 0,1,2.5,5,10 \rbrace$. In all simulations we fixed $\sigma_0 =1.$
Figure \ref{fig:Solb-1d-1} depicts the numerical solutions of the bacterial density of $b$ at different time steps and different values of $n_0$ and $\chi_0$. It should be noted that solutions of $b$ initially do not exhibit a traveling wave profile, nevertheless, eventually the solution resembles a traveling wave. Therefore, the traveling wave estimations were computed at later times by selecting a point in the vicinity of 0.5. These are the points shown in figure \ref{fig:Solb-1d-1}.
In figure \ref{fig:Vel-1} we present a comparison of the propagation velocity of the front between the numerical estimations and the velocity threshold given by equations (\ref{c-bar-chi0}) for $\chi_0>0$ or (\ref{c-bar-0}) for $\chi_0=0$, as functions of the initial nutrient $n_0$. 
For the case without chemotaxis, we found sharper results (see figure \ref{fig:Vel-xi=0}) than those in the original work by Kawasaki \textit{et al.}, because we are considering a more accurate velocity threshold, as was explained in section \ref{secaproxspeed}.
As expected from equation \eqref{speedinq}, when a chemotactic signal is present the velocity of the front is greatly increased. We can see that the numerical velocity fits the velocity threshold.
Therefore, figure \ref{fig:Vel-1} highlights two facts: first, numerical estimations of the propagation speed of the one-dimensional system are in good accordance with the theoretical speed thresholds that are defined for the scalar equation for $b$ (see equation \eqref{scalar-eq-b-re}); and second, the conservation-like equation \eqref{mass-conser} is a good approximation of the solutions of system \eqref{velo-system}.

\begin{figure}[ht!]
\begin{center}
\subfigure[$n_0=0.71$, $\chi_0=10$]{
            \label{fig:Solb-1d-1a}
            \includegraphics[scale=0.38]{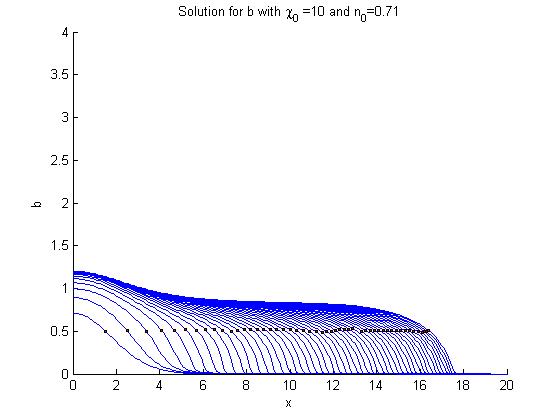}
        }
\subfigure[$n_0=2$, $\chi_0=2.5$]{
            \label{fig:Solb-1d-1b}
            \includegraphics[scale=0.38]{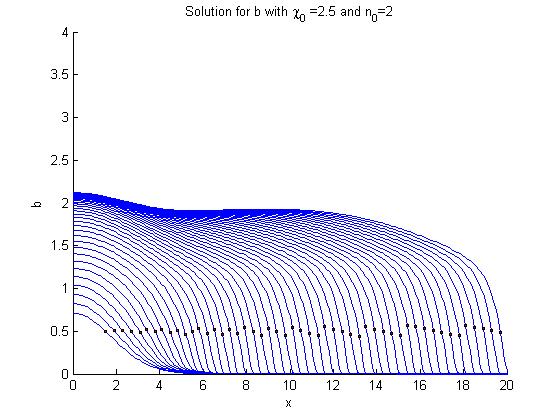}
        }\\
\subfigure[ $n_0=3$, $\chi_0=1$]{
            \label{fig:Solb-1d-2a}
            \includegraphics[scale=0.38]{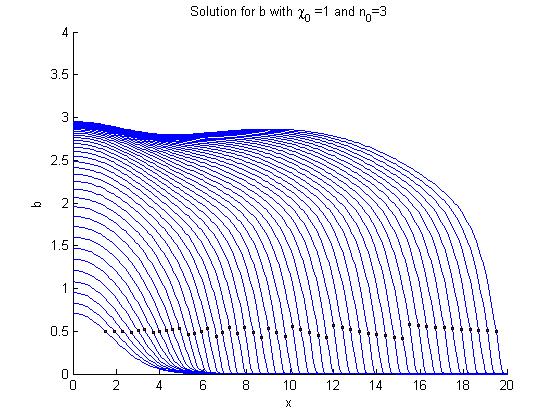}
        }
\subfigure[$n_0=4$, $\chi_0=1$]{
            \label{fig:Solb-1d-2b}
            \includegraphics[scale=0.38]{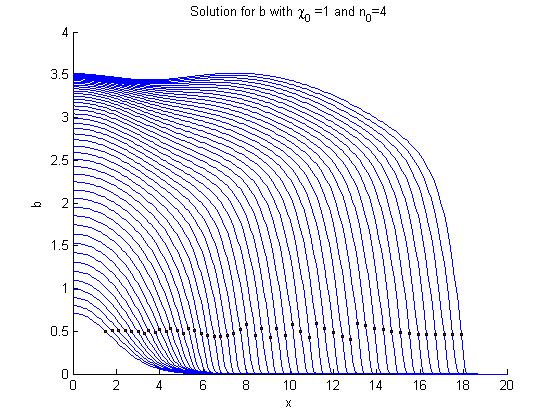}
        }\\
\end{center}
\caption{Numerical solution of the bacterial density $b$ from the system (\ref{velo-system-1d}) 
at different times and different values of $n_0$ and $\chi_0$.}
\label{fig:Solb-1d-1}
\end{figure}

\begin{figure}[ht!]
\begin{center}
\subfigure[$\chi_0 = 0$]{
            \label{fig:Vel-xi=0}
            \includegraphics[scale=0.375]{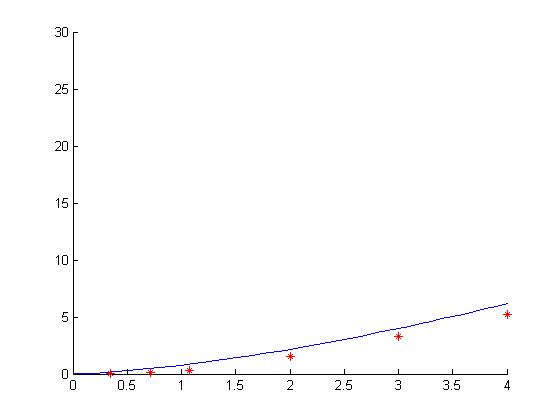}
        }
\subfigure[$\chi_0=1.0$]{
            \label{fig:Vel-xi=1}
            \includegraphics[scale=0.375]{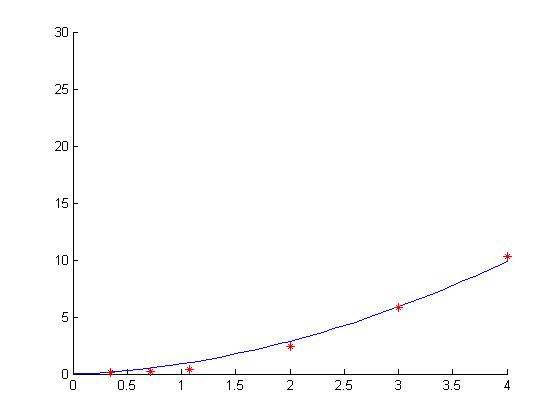}
        }\\
\subfigure[$\chi_0 =2.5$]{
            \label{fig:Vel-xi=2.5}
            \includegraphics[scale=0.375]{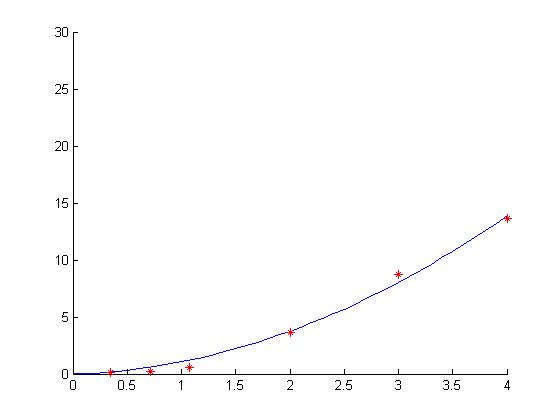}
        }
\subfigure[$\chi_0=5.0$]{
            \label{fig:Vel-xi=5}
            \includegraphics[scale=0.375]{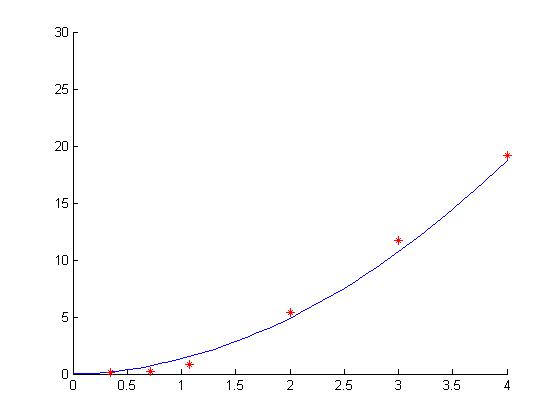}
        }\\
\subfigure[$\chi_0=10.0$]{
            \label{fig:Vel-xi=10}
            \includegraphics[scale=0.375]{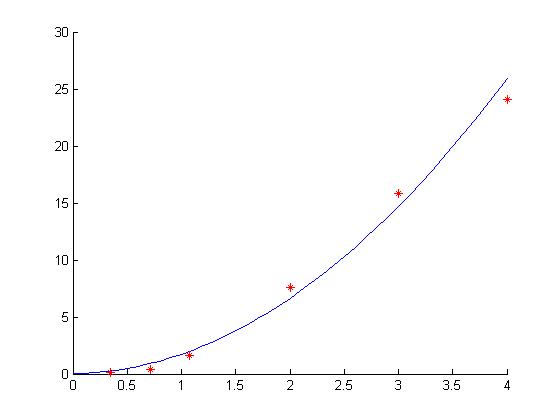}
        }
\end{center}

\caption{Comparison of the propagation speeds of the front for the one-dimensional system (\ref{velo-system-1d}) as a function of the initial nutrient concentration for different values of the chemotactic sensitivity $\chi_0$. The asterisks depict the numerical estimations of the velocity. The solid line is the velocity threshold given by equation \eqref{c-bar-0} in the absence of chemotaxis (case \ref{fig:Vel-xi=0}), and by equation \eqref{c-bar-chi0} when chemotaxis is switched on (cases \ref{fig:Vel-xi=1}, \ref{fig:Vel-xi=2.5}, \ref{fig:Vel-xi=5} and \ref{fig:Vel-xi=10}), both as functions of the initial nutrient concentration $n_0$. Here $\sigma=1$. }
\label{fig:Vel-1}
\end{figure}

\begin{figure}[ht!]
\begin{center}
\subfigure[$\chi_0 = 0$, $t = 694.5$]{
            \label{set1a}
            \includegraphics[scale=0.125]{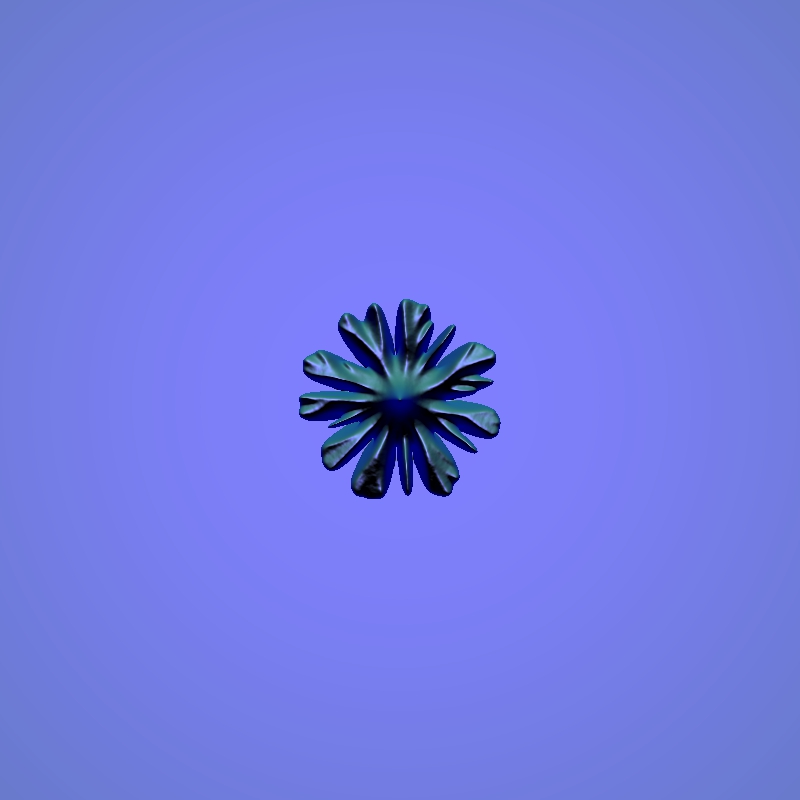}
        }
\subfigure[$\chi_0= 5$, $t = 694.5$]{
            \label{set1b}
            \includegraphics[scale=0.125]{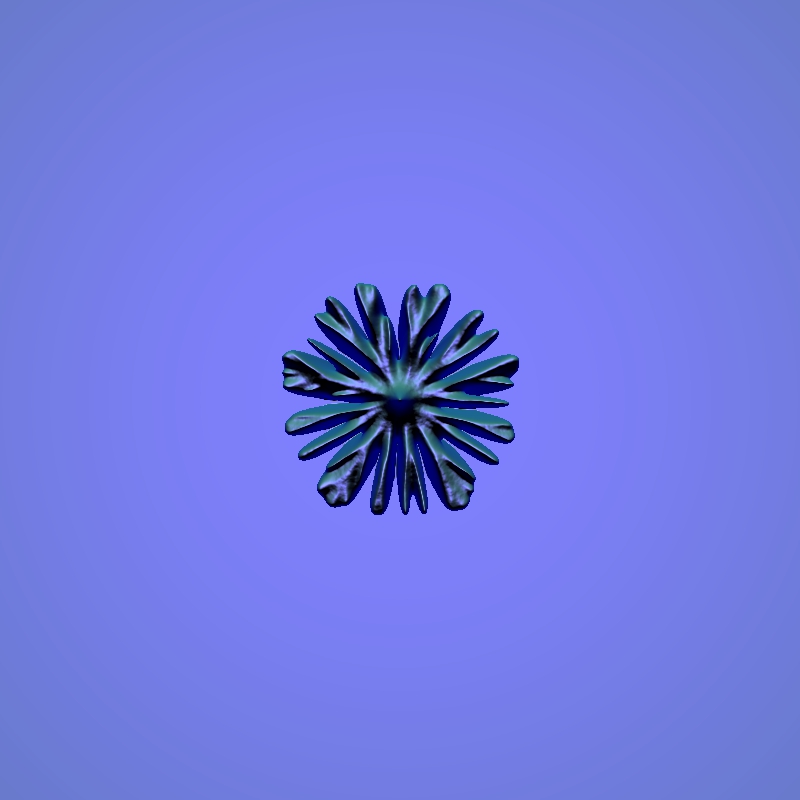}
        }\\
\subfigure[$\chi_0 = 0$, $t = 1389.0$]{
            \label{set1c}
            \includegraphics[scale=0.125]{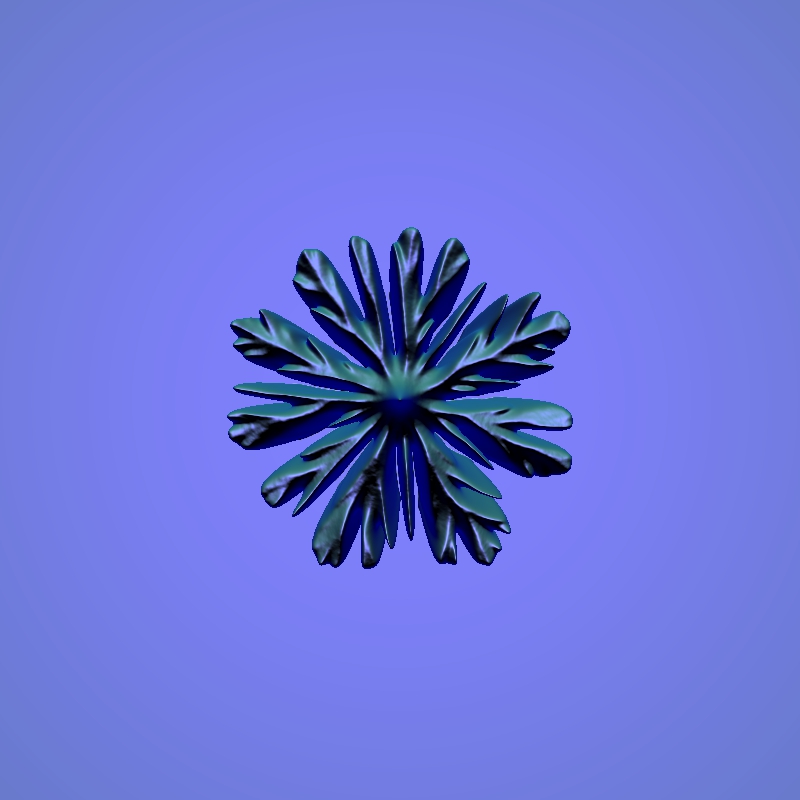}
        }
\subfigure[$\chi_0= 5$, $t = 1389.0$]{
            \label{set1d}
            \includegraphics[scale=0.125]{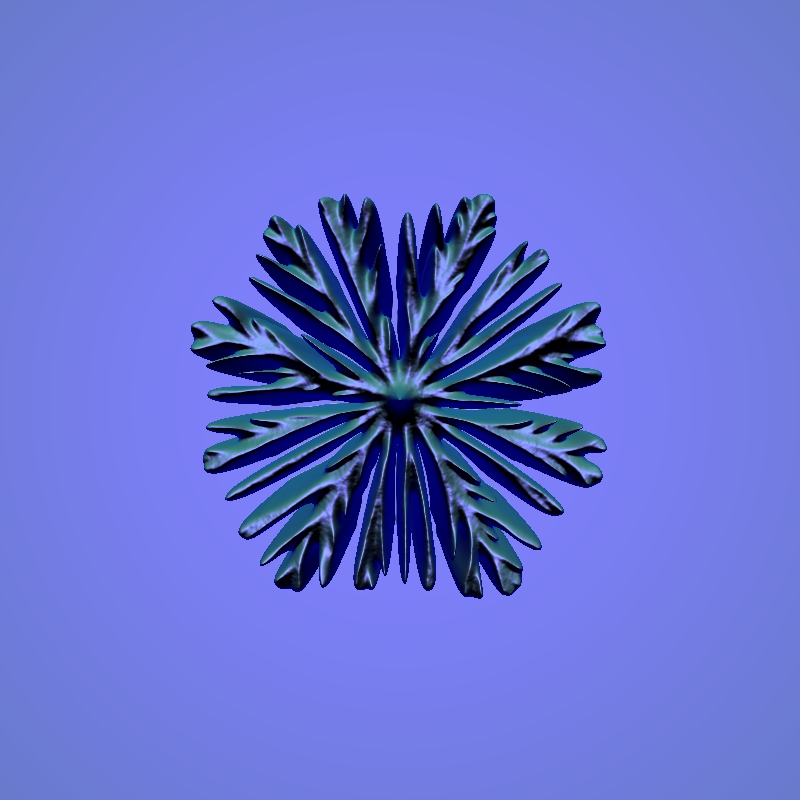}
        }\\
\subfigure[$\chi_0 = 0$, $t = 2083.5$]{
            \label{set1e}
            \includegraphics[scale=0.125]{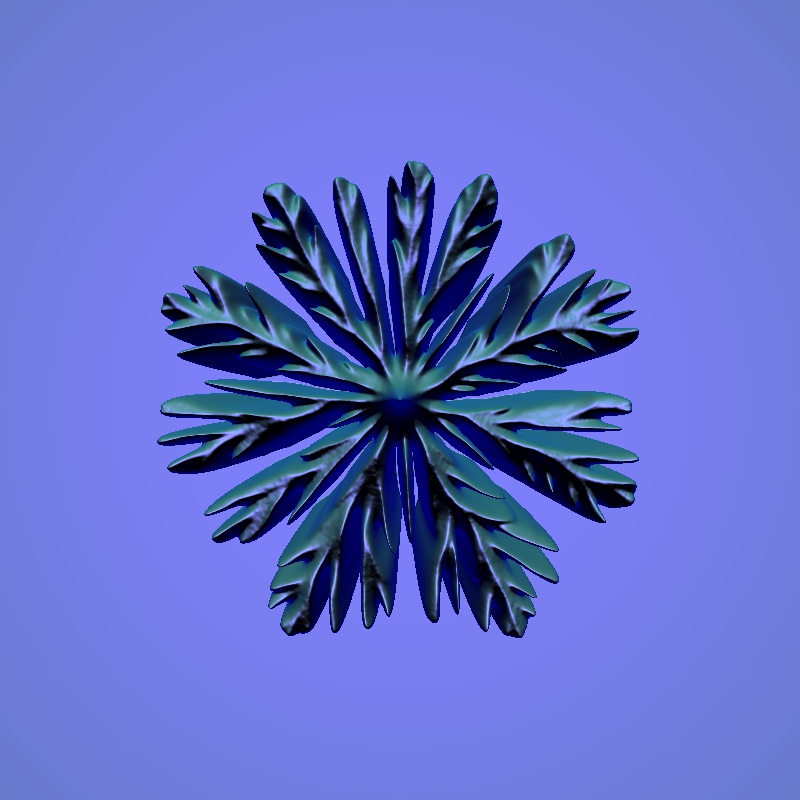}
        }
\subfigure[$\chi_0= 5$, $t = 2083.5$]{
            \label{set1f}
            \includegraphics[scale=0.125]{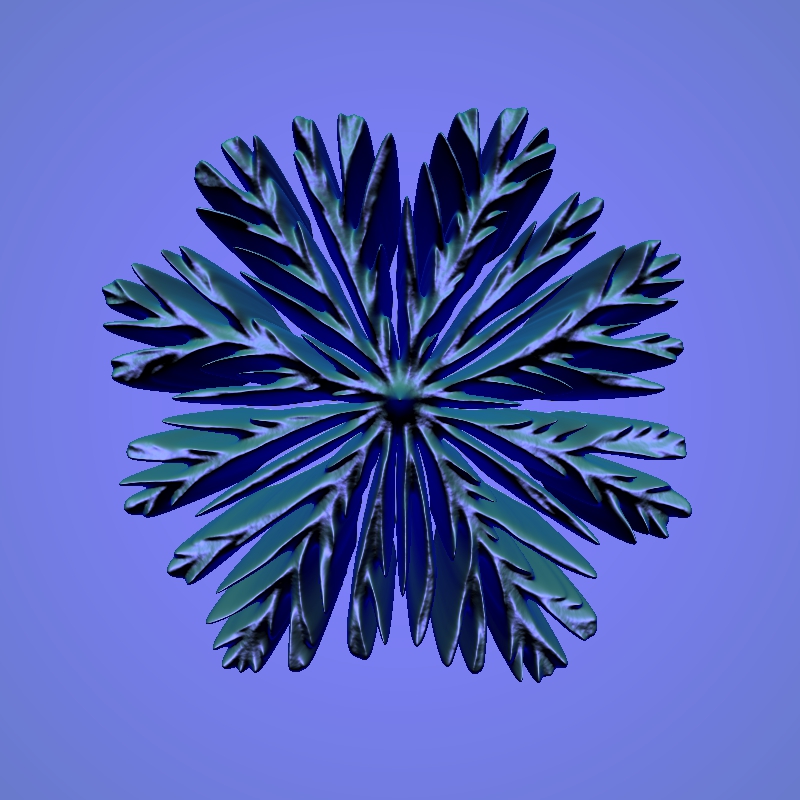}
        }\\
\subfigure[$\chi_0 = 0$, $t = 2824.0$]{
            \label{set1g}
            \includegraphics[scale=0.125]{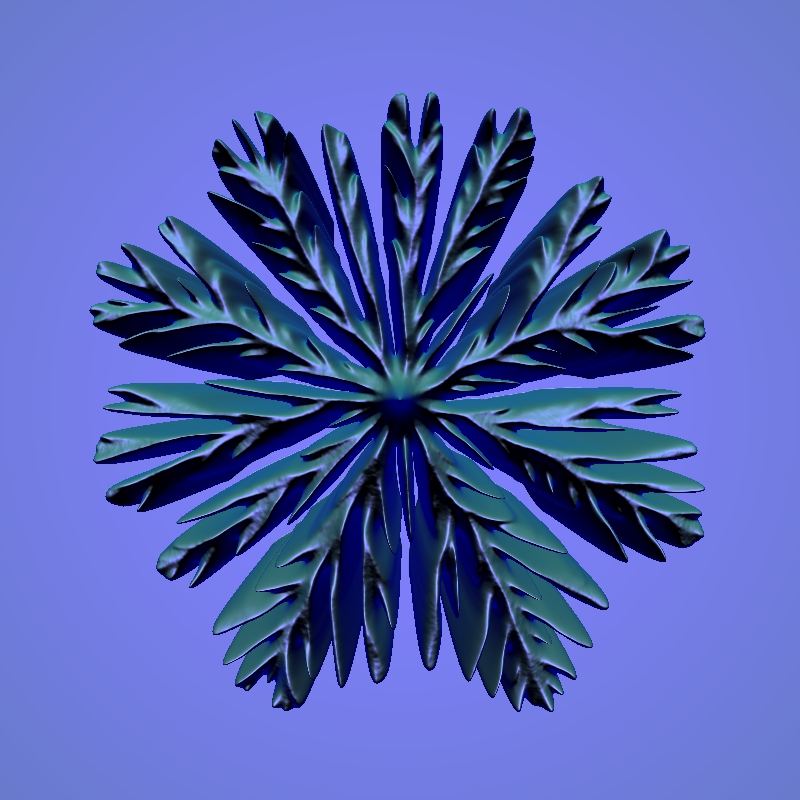}
        }
\subfigure[$\chi_0= 5$, $t = 2824.0$]{
            \label{set1h}
            \includegraphics[scale=0.125]{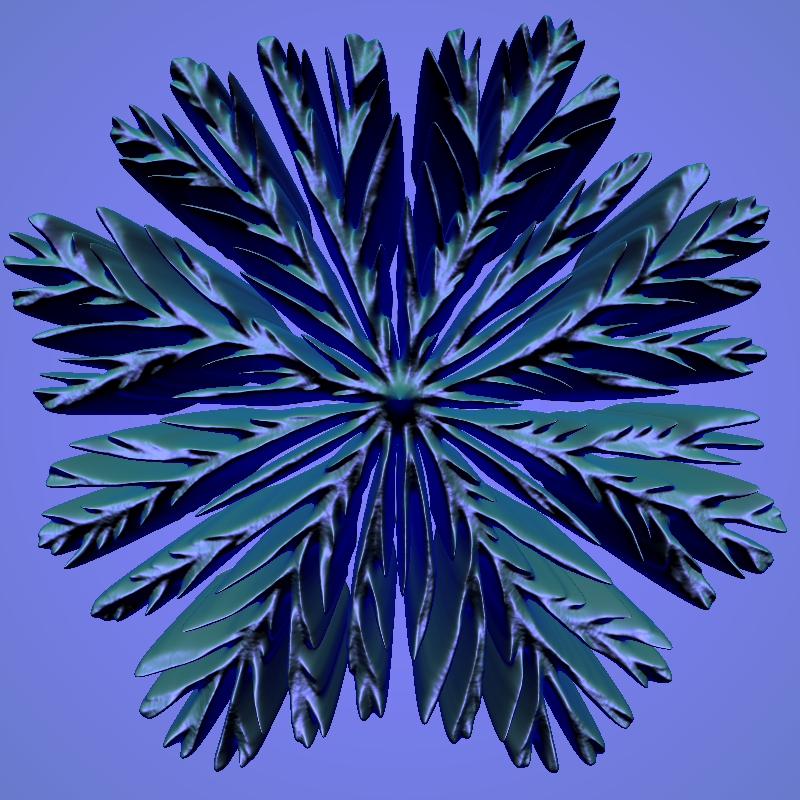}
        }\\
\end{center}
\caption{Colony growth as a result of simulations of system \eqref{ka-chem-model}, taking $\sigma_0 = 1.0, n_0 = 0.71$, with chemotaxis sensitivity $\chi_0 = 0$ (no chemotaxis - left), and $\chi_0 = 5.0$ (right), for different values of $t$.
}
\label{fig:set1}
\end{figure}

\begin{figure}[ht!]
\begin{center}
\subfigure[$\chi_0 = 0$, $t = 115.7$]{
            \label{set2a}
            \includegraphics[scale=0.125]{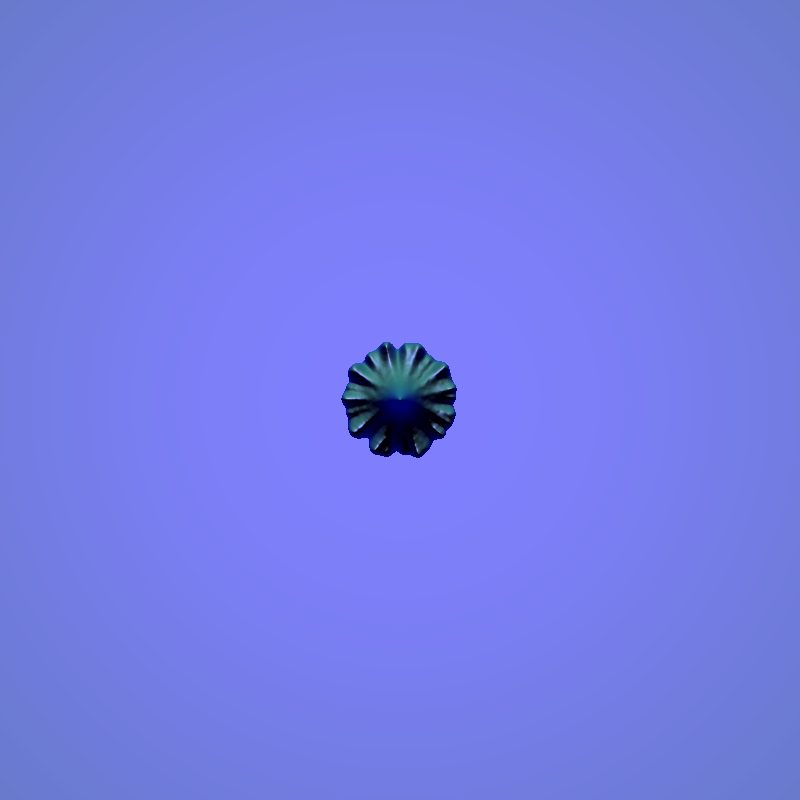}
        }
\subfigure[$\chi_0=5$, $t = 115.7$]{
            \label{set2b}
            \includegraphics[scale=0.125]{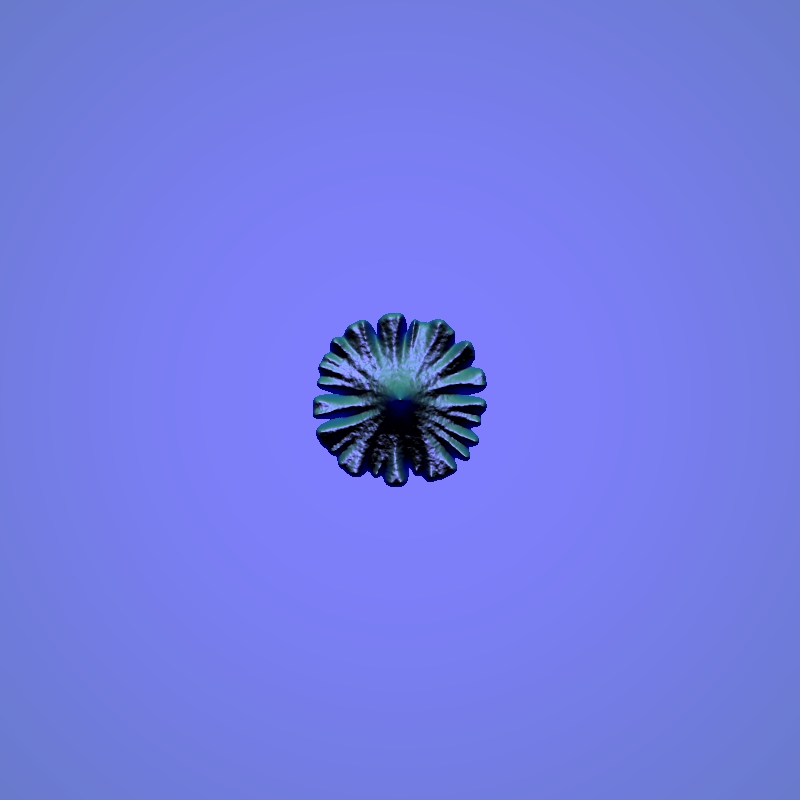}
        }\\
\subfigure[$\chi_0 = 0$, $t = 231.5$]{
            \label{set2c}
            \includegraphics[scale=0.125]{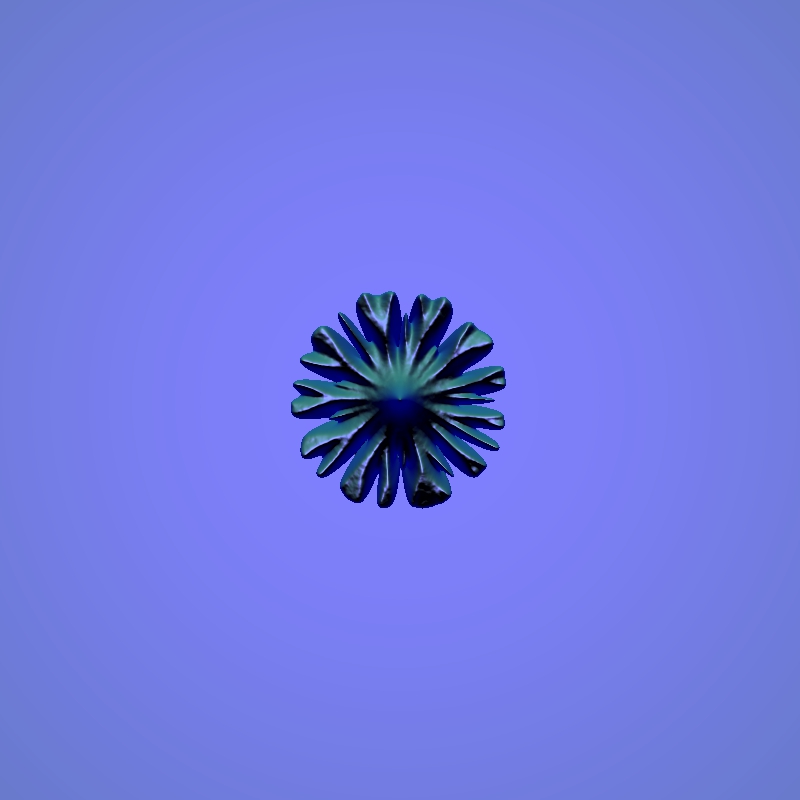}
        }
\subfigure[$\chi_0=5$, $t = 231.5$]{
            \label{set2d}
            \includegraphics[scale=0.125]{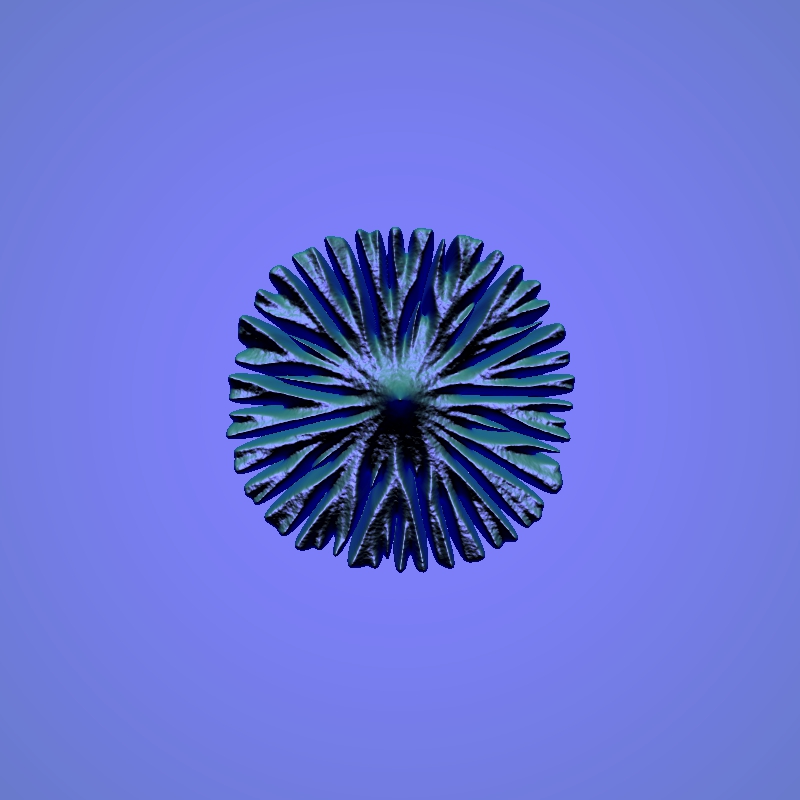}
        }\\
\subfigure[$\chi_0 = 0$, $t = 370.4$]{
            \label{set2e}
            \includegraphics[scale=0.125]{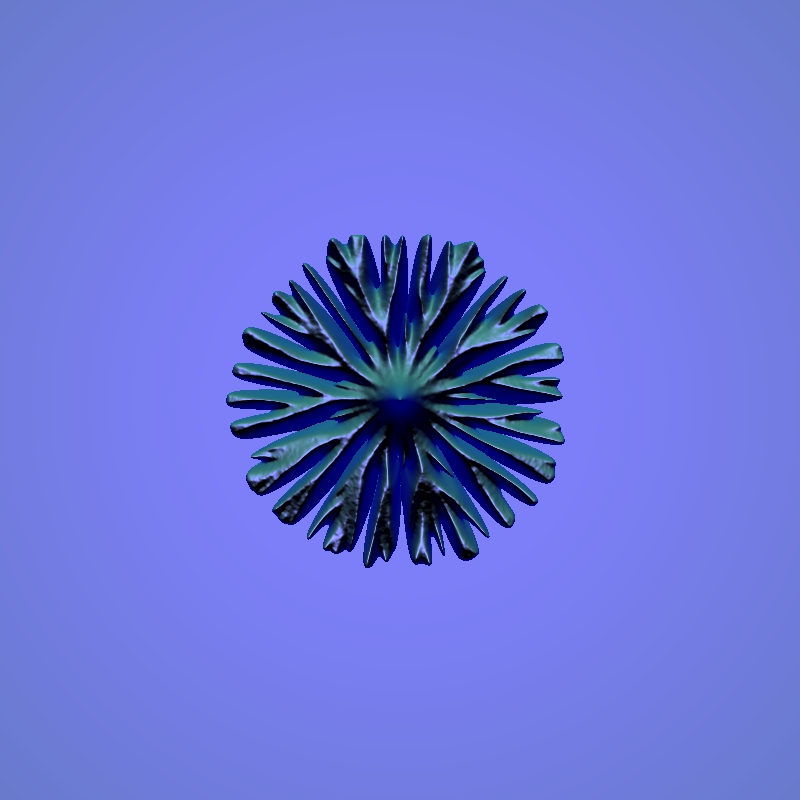}
        }
\subfigure[$\chi_0=5$, $t = 370.4$]{
            \label{set2f}
            \includegraphics[scale=0.125]{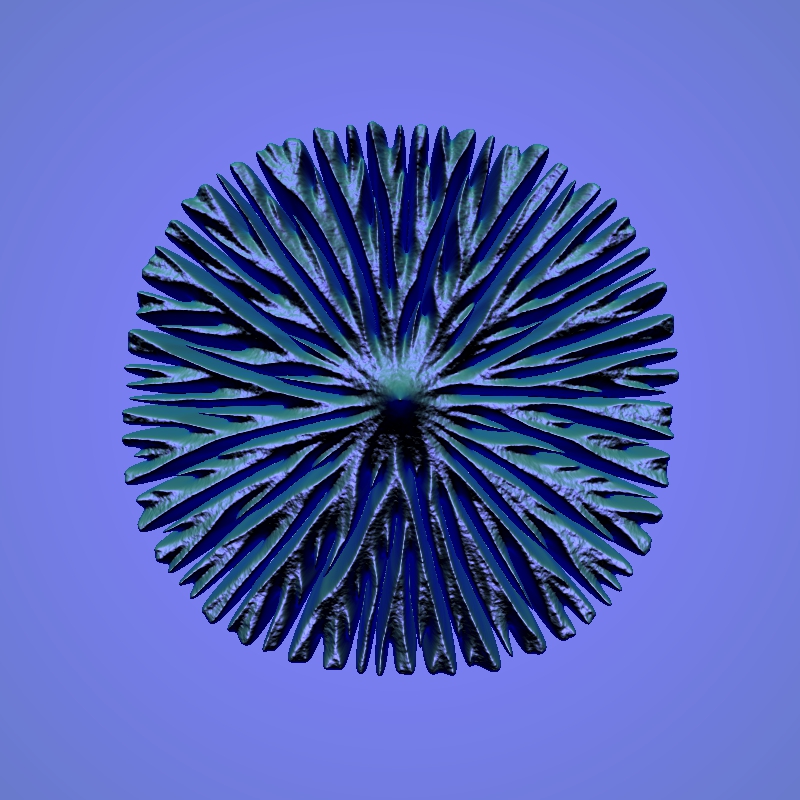}
        }\\
\subfigure[$\chi_0 = 0$, $t = 509.3$]{
            \label{set2g}
            \includegraphics[scale=0.125]{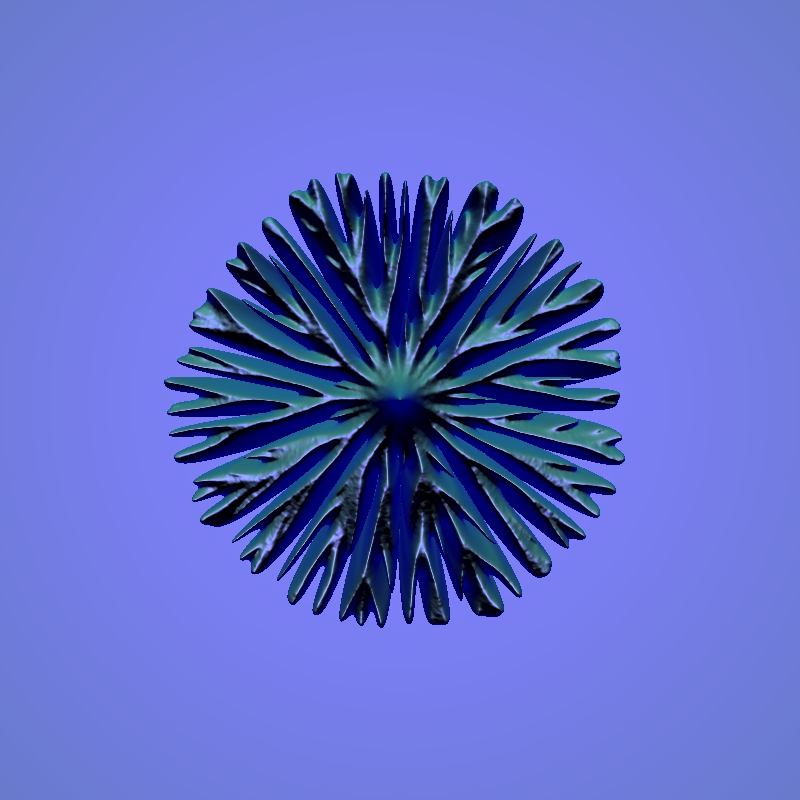}
        }
\subfigure[$\chi_0=5$, $t = 509.3$]{
            \label{set2h}
            \includegraphics[scale=0.125]{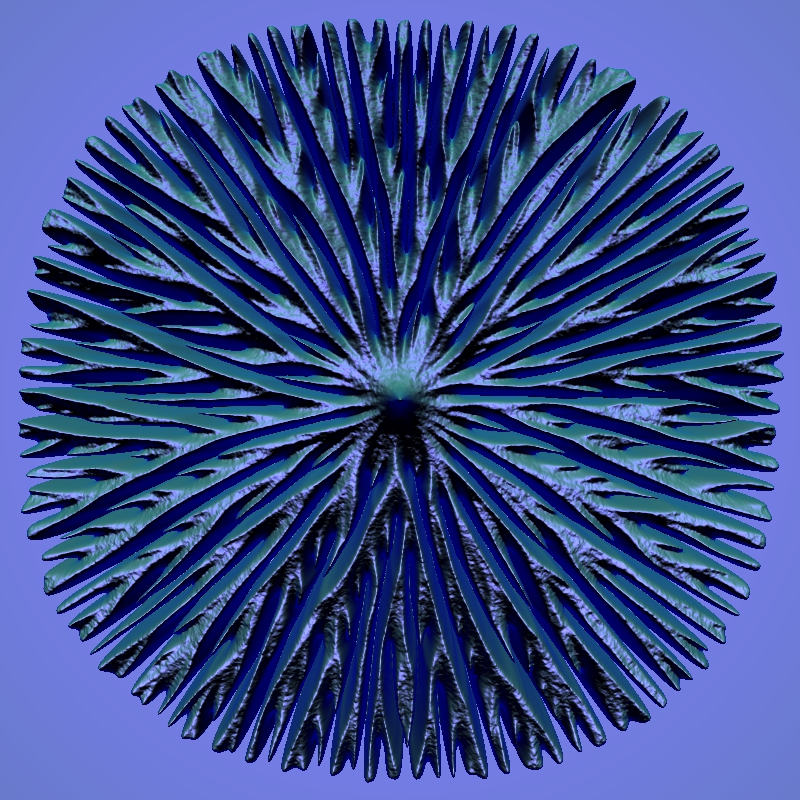}
        }\\
\end{center}
\caption{Colony growth as a result of simulations of system \eqref{ka-chem-model}, taking $\sigma_0 = 1.0, n_0 = 1.07$, with chemotaxis sensitivity $\chi_0 = 0$ (no chemotaxis - left), and $\chi_0 = 5.0$ (right), for different values of $t$.}
\label{fig:set2}
\end{figure}

\begin{figure}[ht!]
\begin{center}
\subfigure[$\chi_0 = 0$, $t = 92.6$]{
            \label{set3aa}
            \includegraphics[scale=0.125]{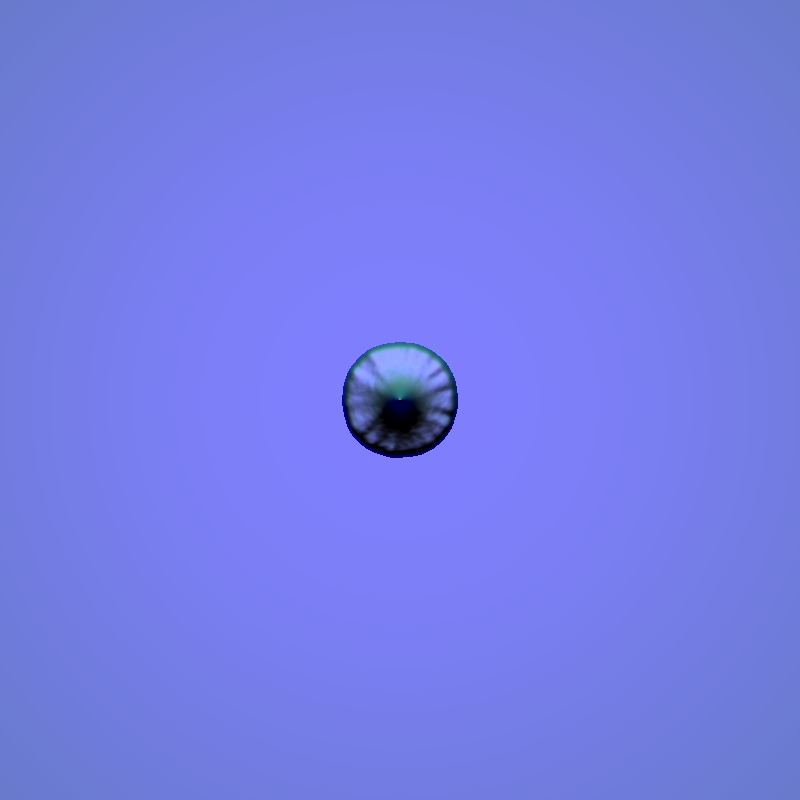}
        }
\subfigure[$\chi_0=2.5$, $t = 92.6$]{
            \label{set3ab}
            \includegraphics[scale=0.125]{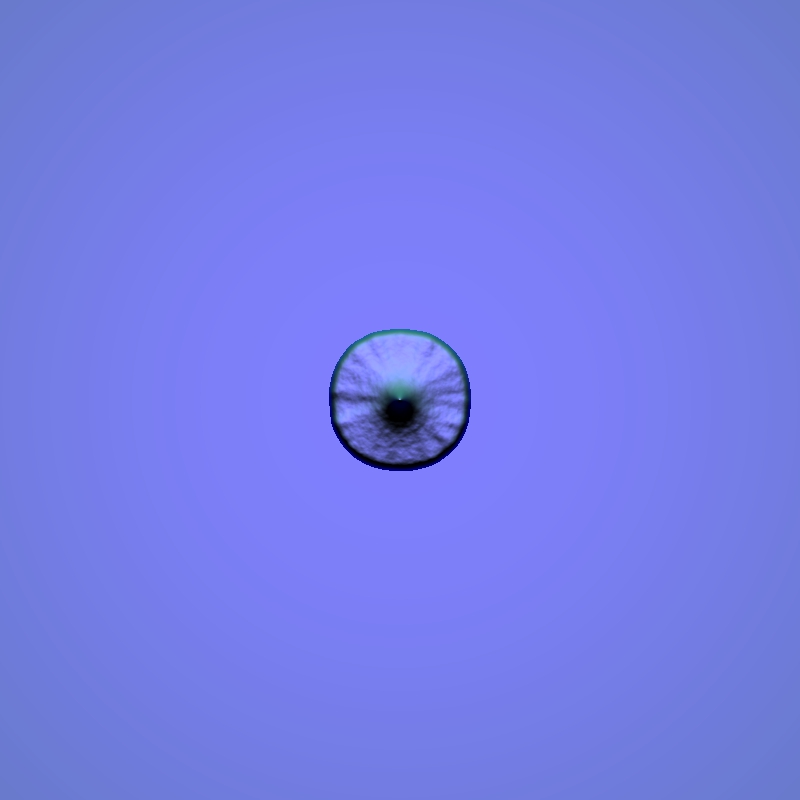}
        }\\
\subfigure[$\chi_0 = 0$, $t = 162.0$]{
            \label{set3ac}
            \includegraphics[scale=0.125]{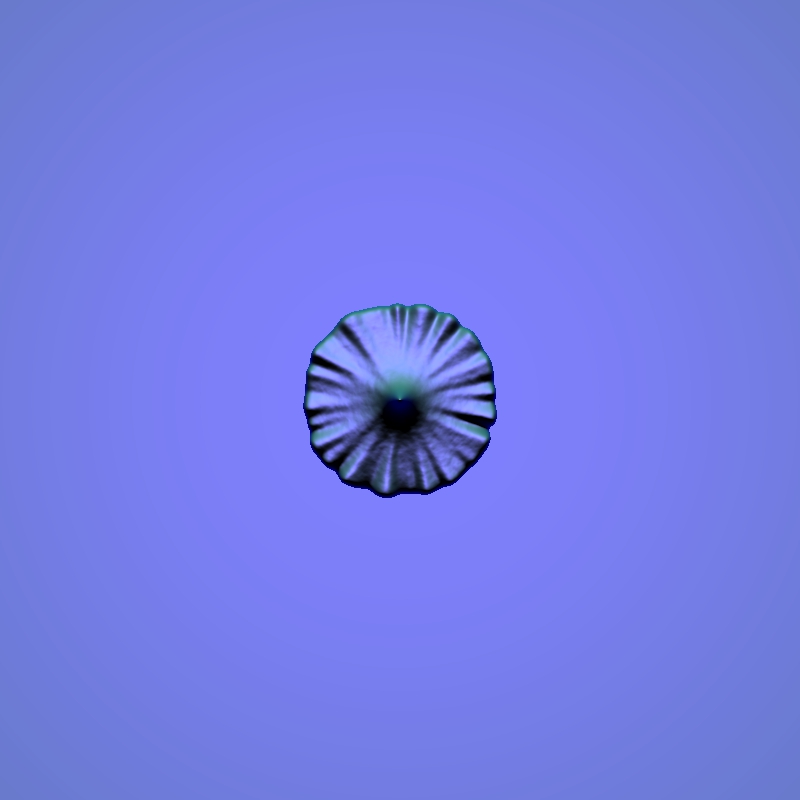}
        }
\subfigure[$\chi_0=2.5$, $t = 162.0$]{
            \label{set3ad}
            \includegraphics[scale=0.125]{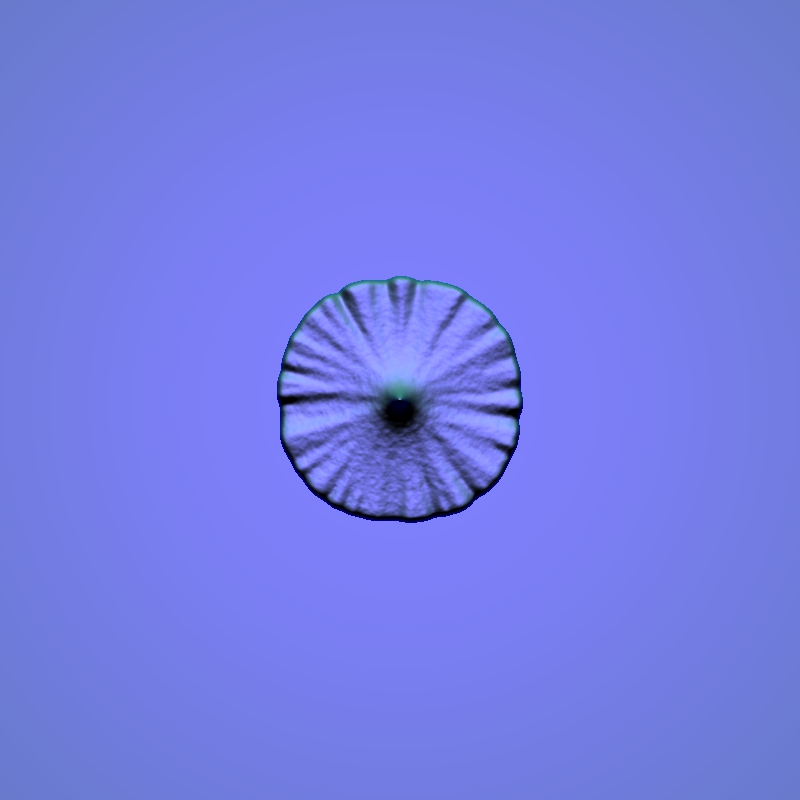}
        }\\
\subfigure[$\chi_0 = 0$, $t = 231.5$]{
            \label{set3ae}
            \includegraphics[scale=0.125]{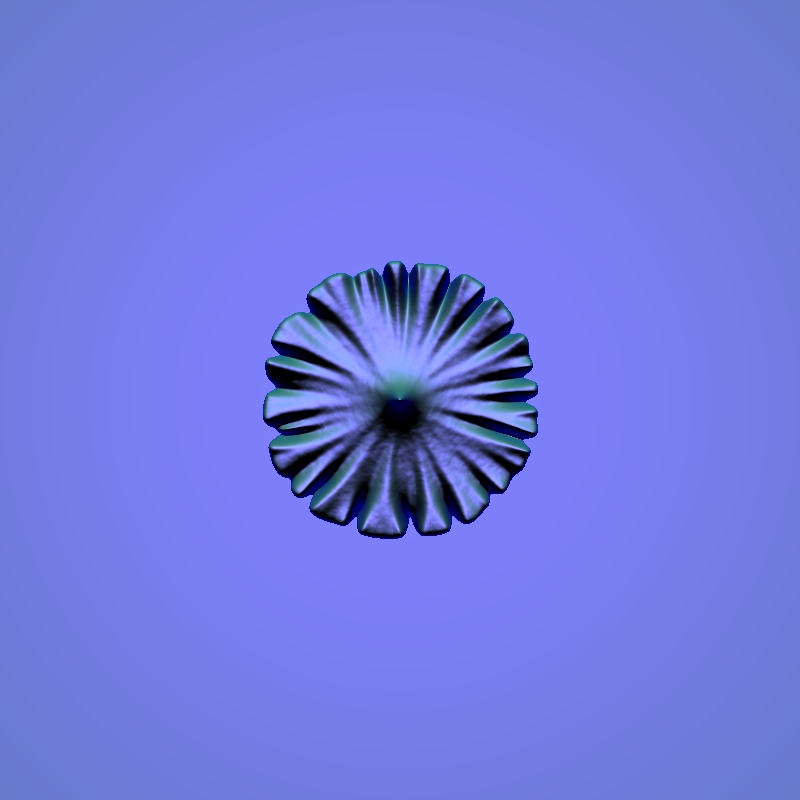}
        }
\subfigure[$\chi_0=2.5$, $t = 231.5$]{
            \label{set3af}
            \includegraphics[scale=0.125]{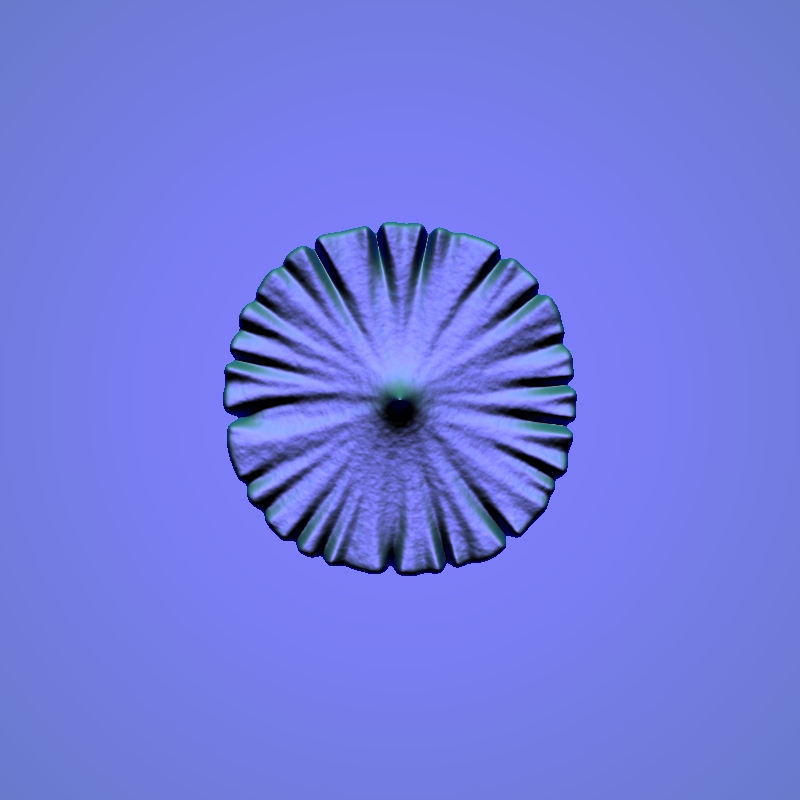}
        }\\
\subfigure[$\chi_0 = 0$, $t = 324.0$]{
            \label{set3ag}
            \includegraphics[scale=0.125]{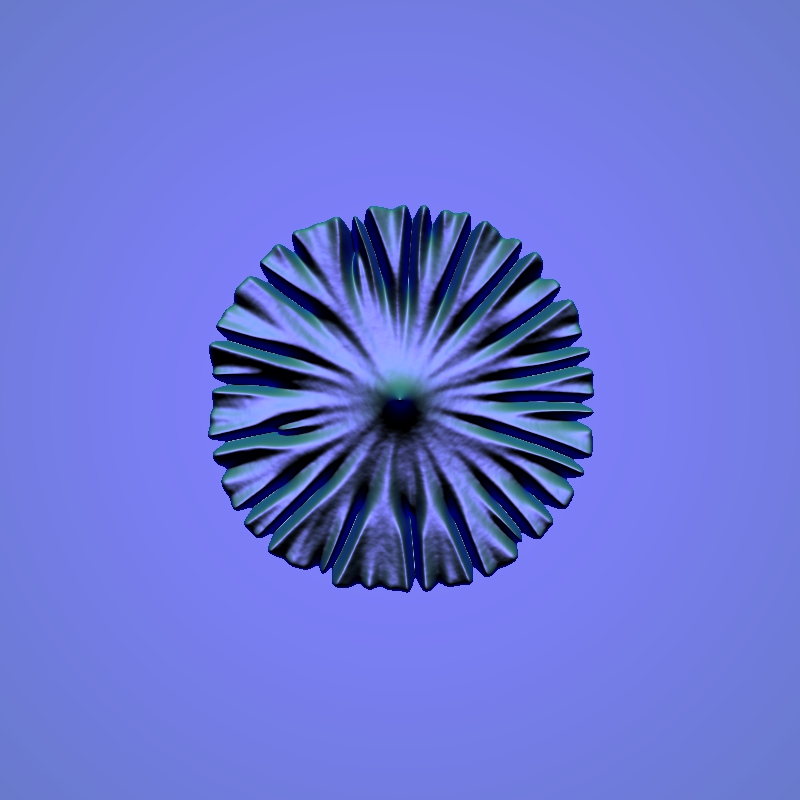}
        }
\subfigure[$\chi_0=2.5$, $t = 324.0$]{
            \label{set3ah}
            \includegraphics[scale=0.125]{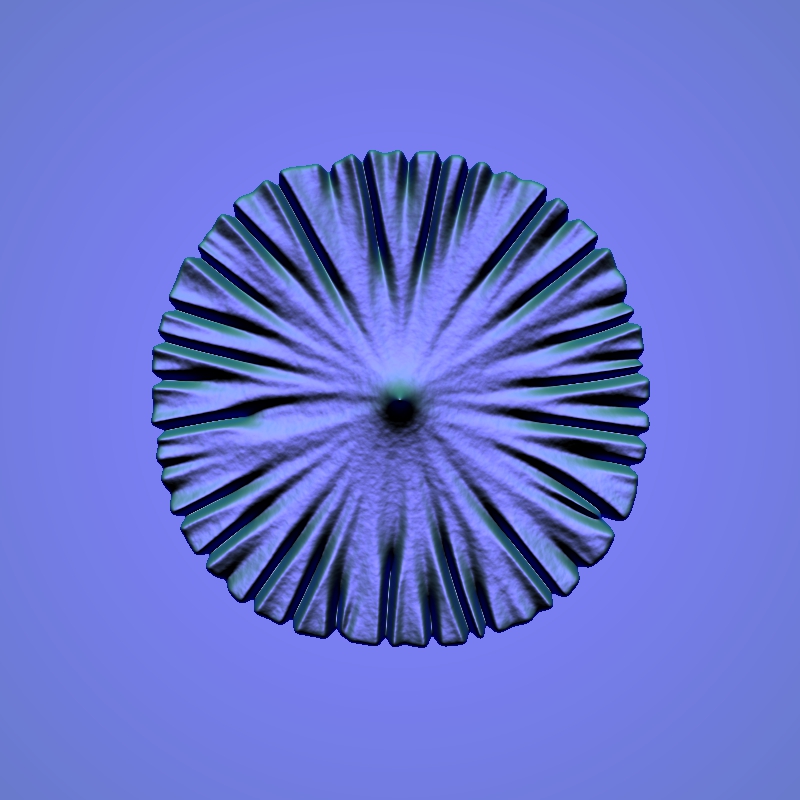}
        }\\
\end{center}
\caption{Colony growth as a result of simulations of system \eqref{ka-chem-model}, taking $\sigma_0 = 4.0, n_0 = 0.71$, with chemotaxis sensitivity $\chi_0 = 0$ (left) and $\chi_0 = 2.5$ (right), for different values of $t$.}
\label{fig:set3a}
\end{figure}

\begin{figure}[ht!]
\begin{center}
\subfigure[$\chi_0 = 5$, $t = 92.6$]{
            \label{set3ba}
            \includegraphics[scale=0.125]{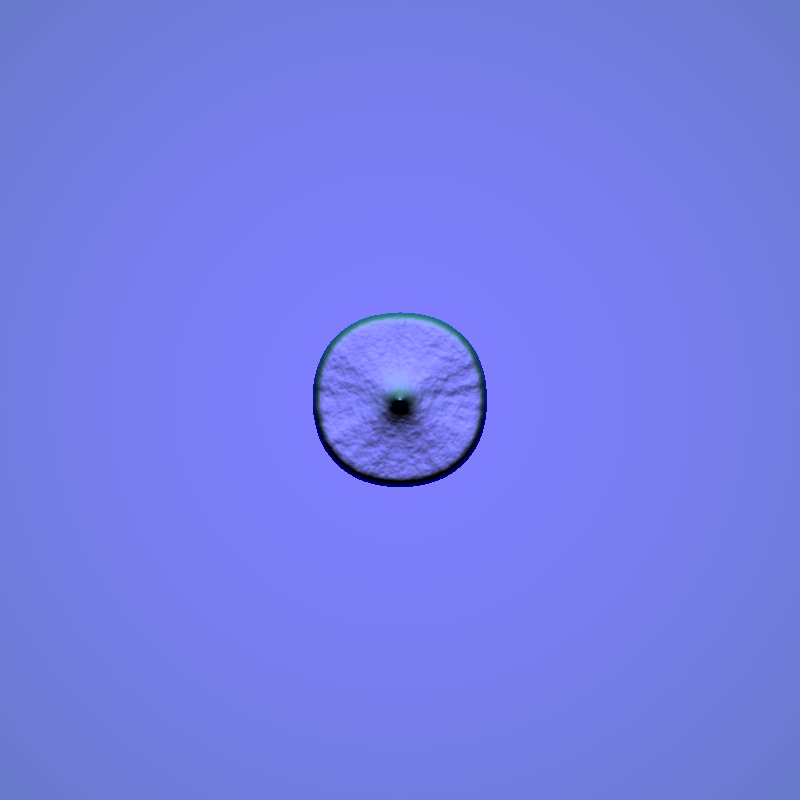}
        }
\subfigure[$\chi_0=7.5$, $t = 92.6$]{
            \label{set3bb}
            \includegraphics[scale=0.125]{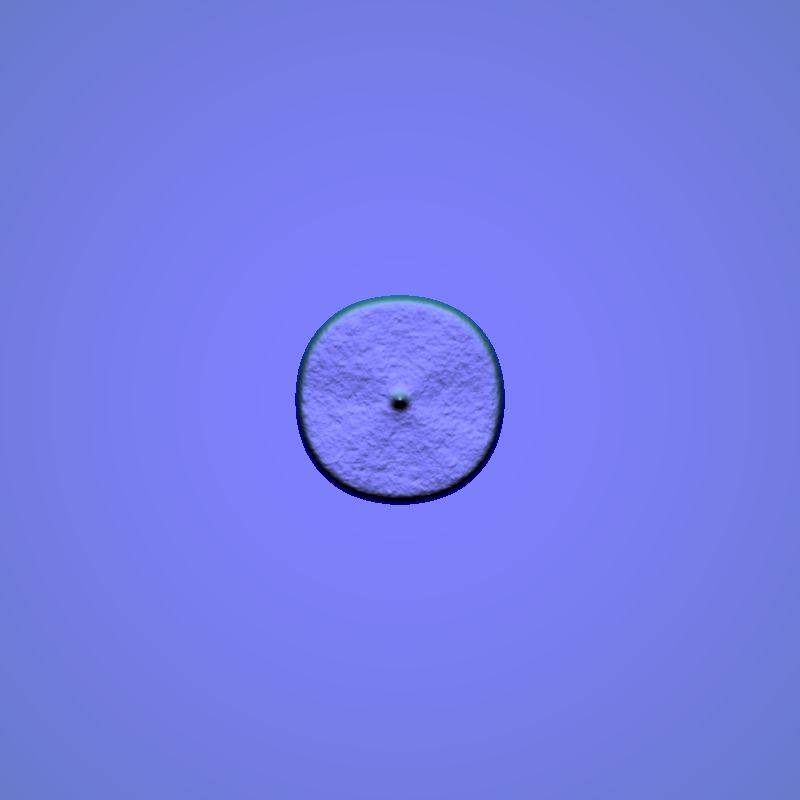}
        }\\
\subfigure[$\chi_0 = 5$, $t = 162.0$]{
            \label{set3bc}
            \includegraphics[scale=0.125]{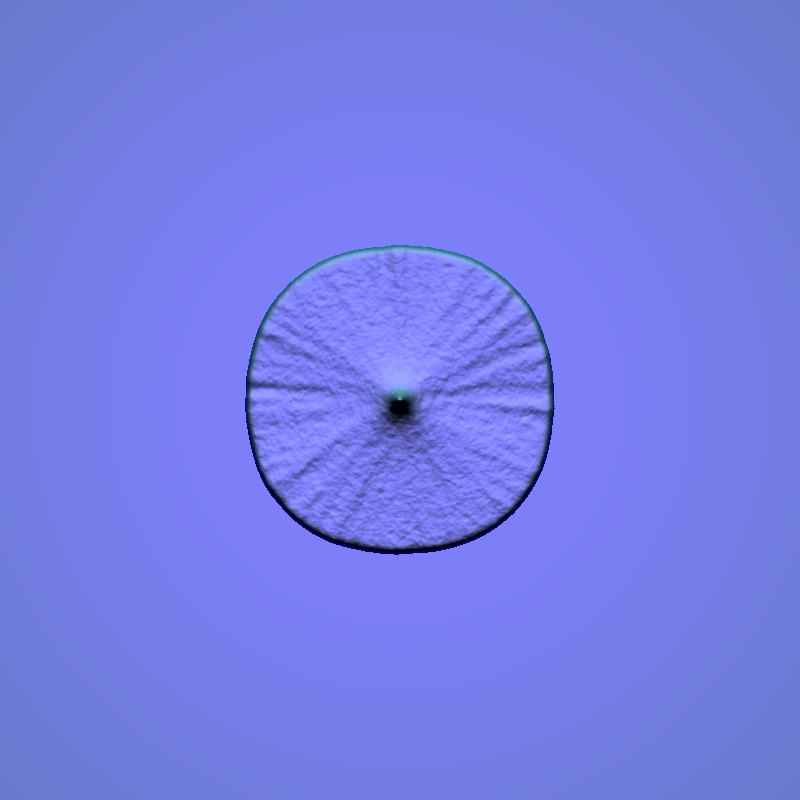}
        }
\subfigure[$\chi_0=7.5$, $t = 162.0$]{
            \label{set3bd}
            \includegraphics[scale=0.125]{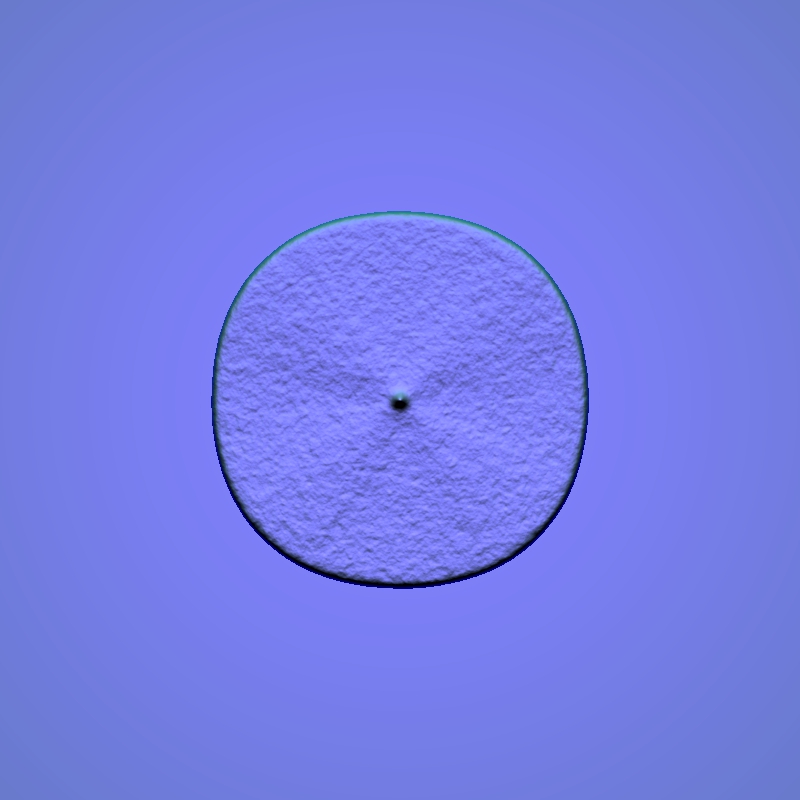}
        }\\
\subfigure[$\chi_0 = 5$, $t = 231.5$]{
            \label{set3be}
            \includegraphics[scale=0.125]{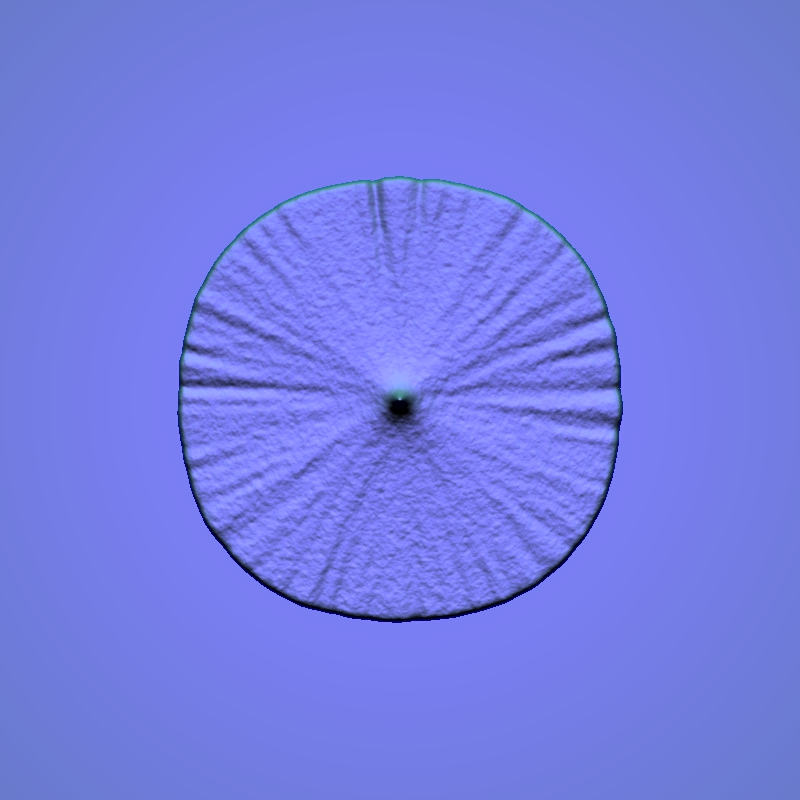}
        }
\subfigure[$\chi_0=7.5$, $t = 231.5$]{
            \label{set3bf}
            \includegraphics[scale=0.125]{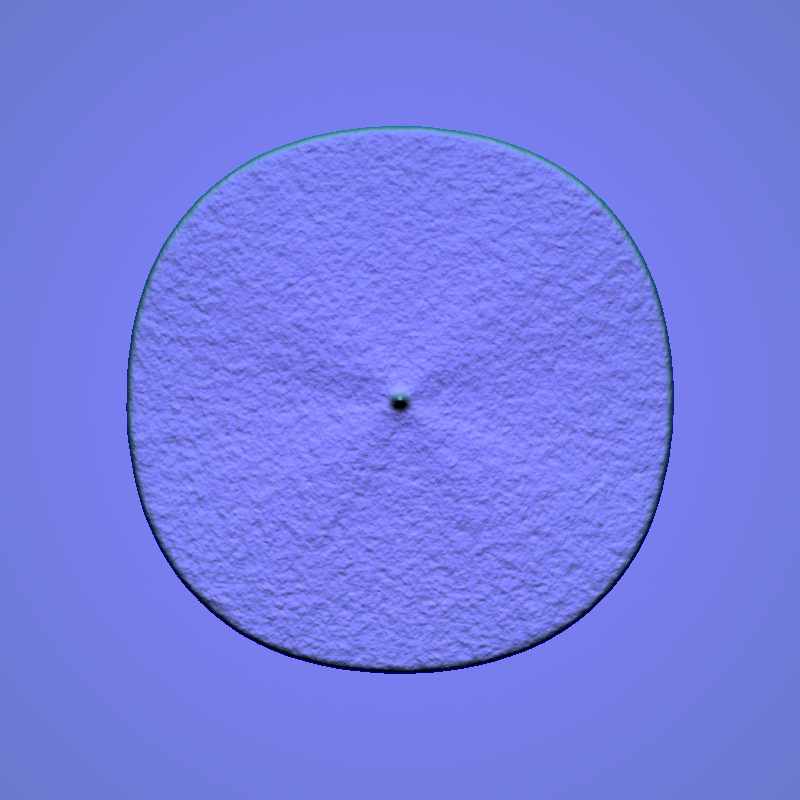}
        }\\
\subfigure[$\chi_0 = 5$, $t = 324.0$]{
            \label{set3bg}
            \includegraphics[scale=0.125]{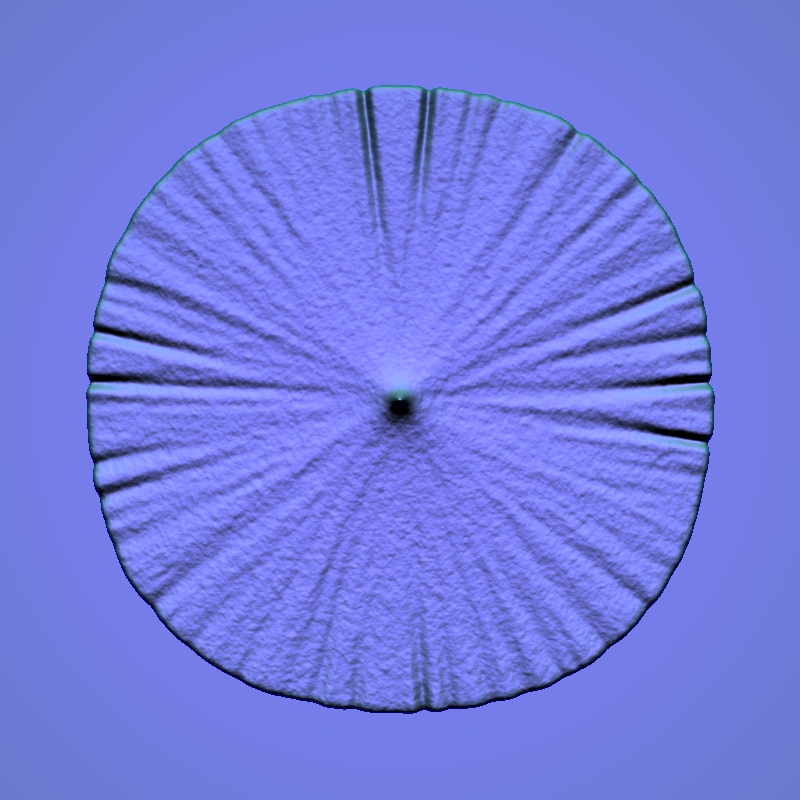}
        }
\subfigure[$\chi_0=7.5$, $t = 324.0$]{
            \label{set3bh}
            \includegraphics[scale=0.125]{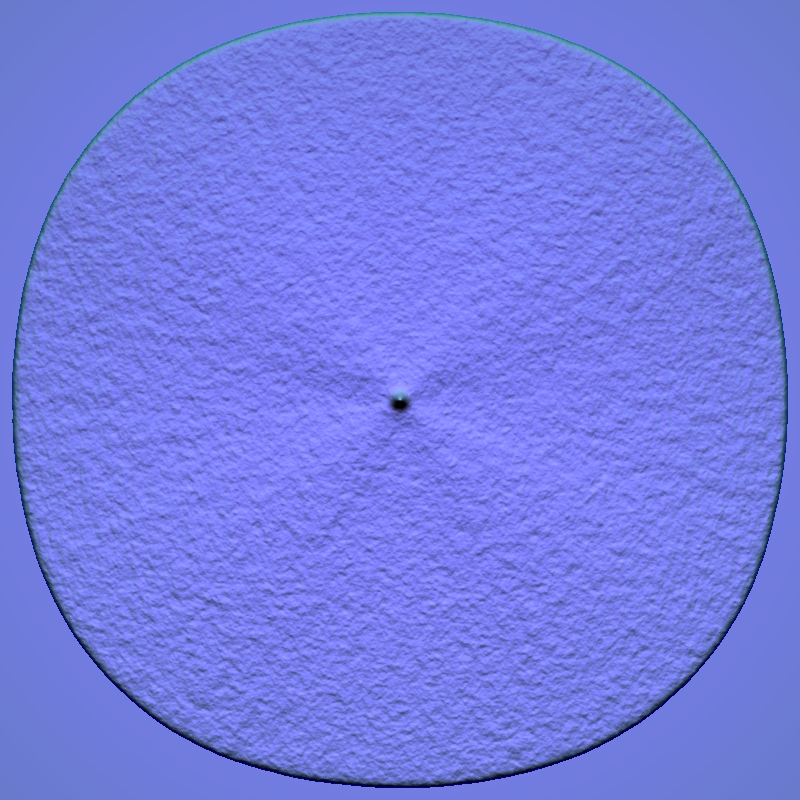}
        }\\
\end{center}
\caption{Colony growth as a result of simulations of system \eqref{ka-chem-model}, taking $\sigma_0 = 4.0, n_0 = 0.71$, with chemotaxis sensitivity $\chi_0 = 5.0$ (left) and $\chi_0 = 7.5$ (right), for different values of $t$.}
\label{fig:set3b}
\end{figure}

\begin{figure}[ht!]
\begin{center}
\subfigure[$\chi_0 = 0$, $t = 33.0$]{
            \label{set4a}
            \includegraphics[scale=0.125]{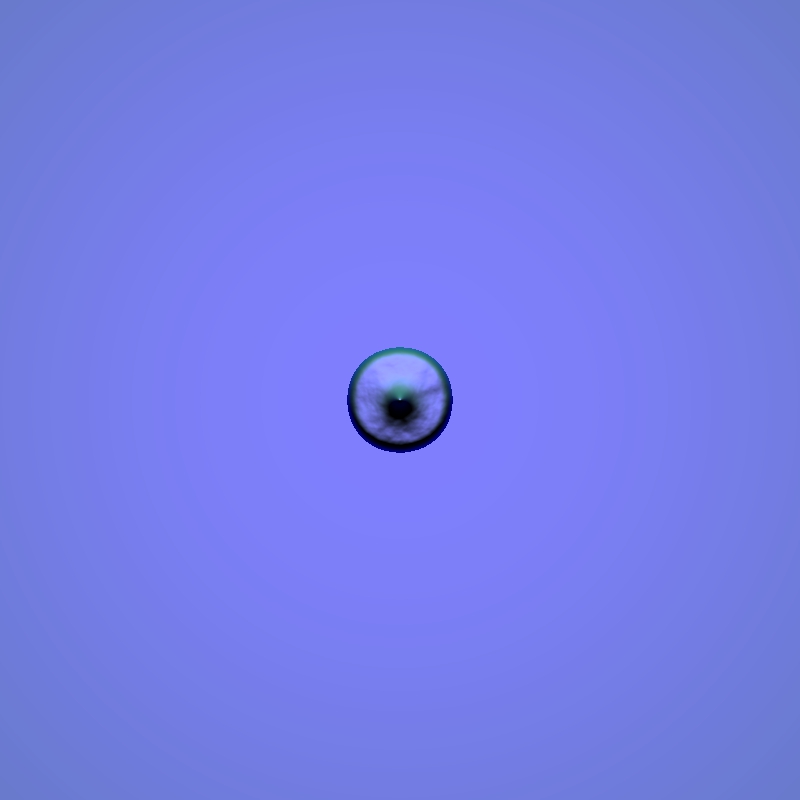}
        }
\subfigure[$\chi_0=7.5$, $t = 33.0$]{
            \label{set4b}
            \includegraphics[scale=0.125]{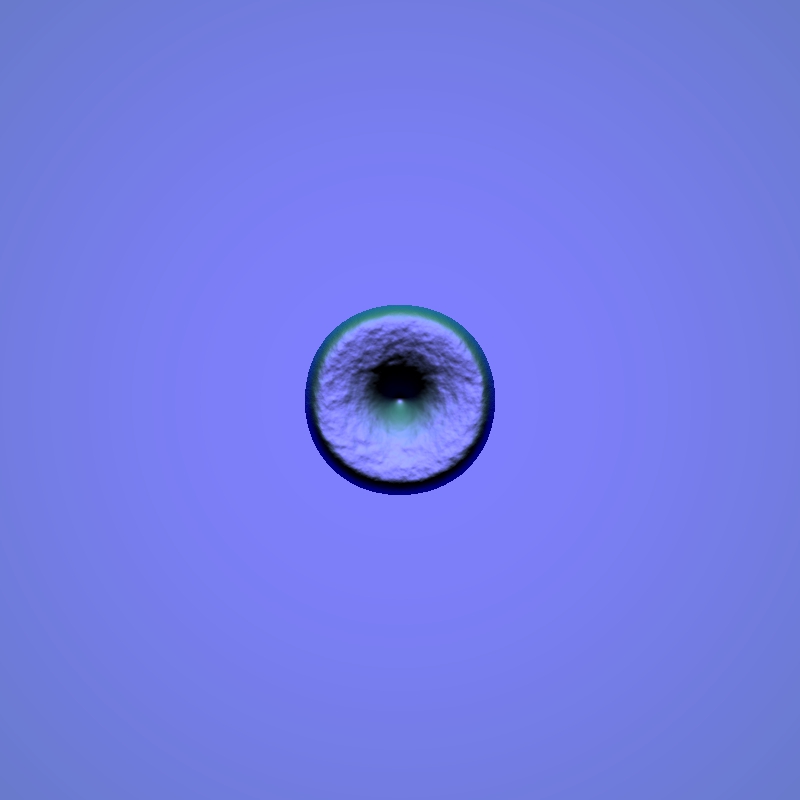}
        }\\
\subfigure[$\chi_0 = 0$, $t = 66.0$]{
            \label{set4c}
            \includegraphics[scale=0.125]{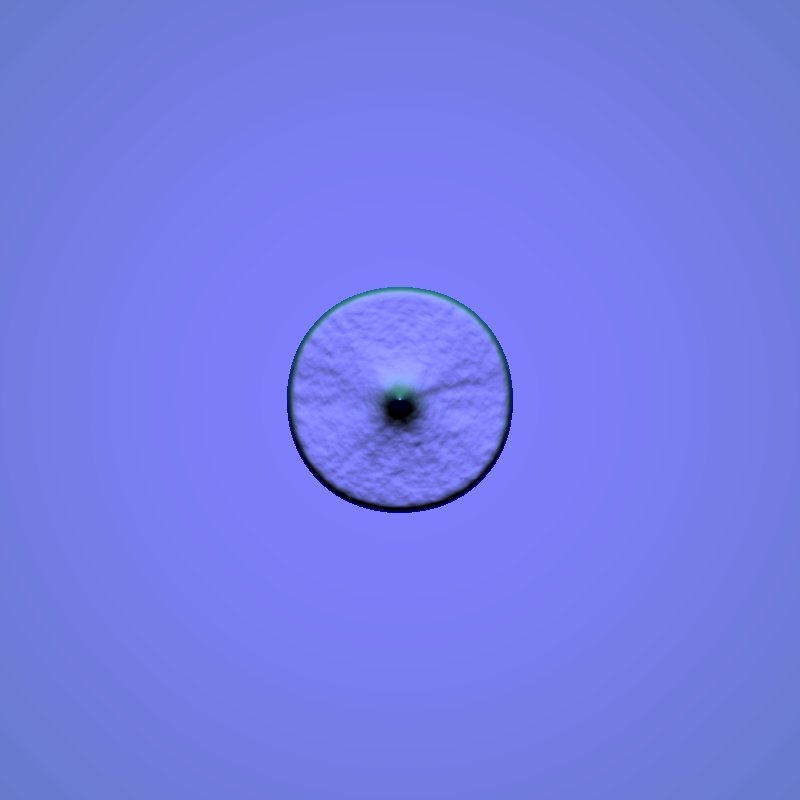}
        }
\subfigure[$\chi_0=7.5$, $t = 66.0$]{
            \label{set4d}
            \includegraphics[scale=0.125]{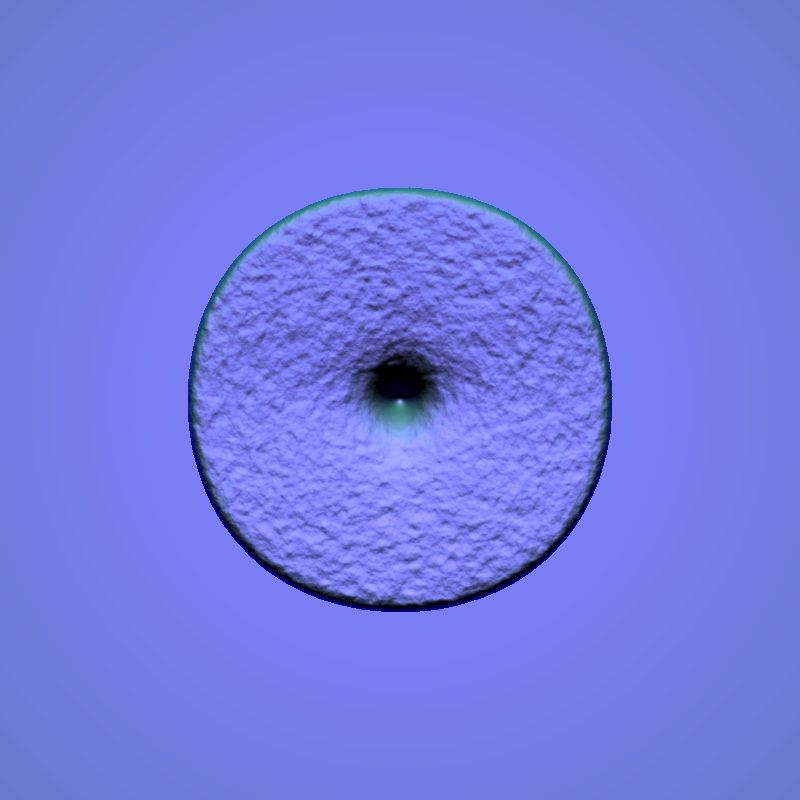}
        }\\
\subfigure[$\chi_0 = 0$, $t = 88.0$]{
            \label{set4e}
            \includegraphics[scale=0.125]{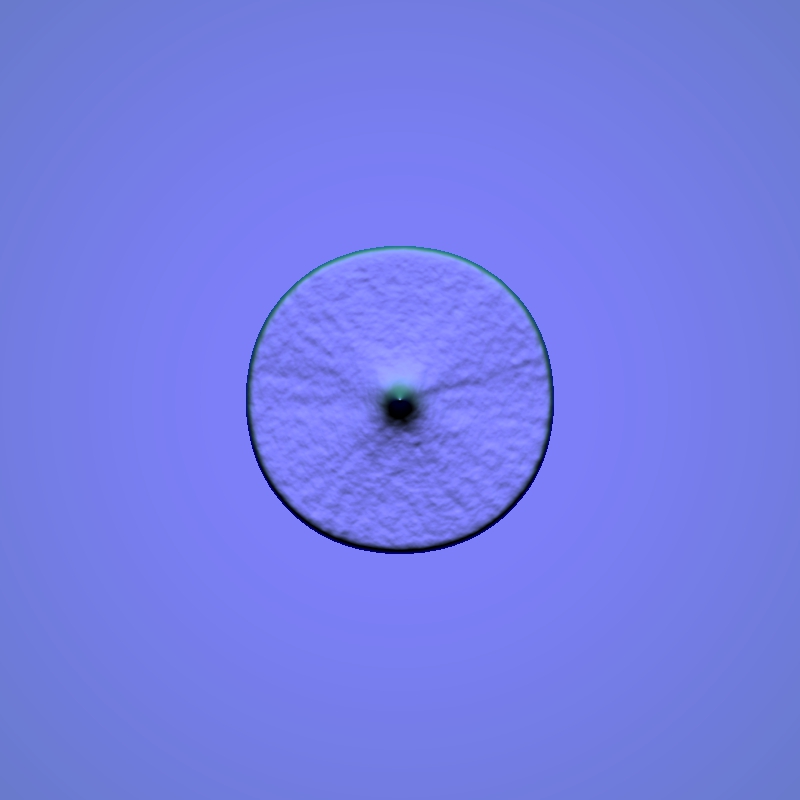}
        }
\subfigure[$\chi_0=7.5$, $t = 88.0$]{
            \label{set4f}
            \includegraphics[scale=0.125]{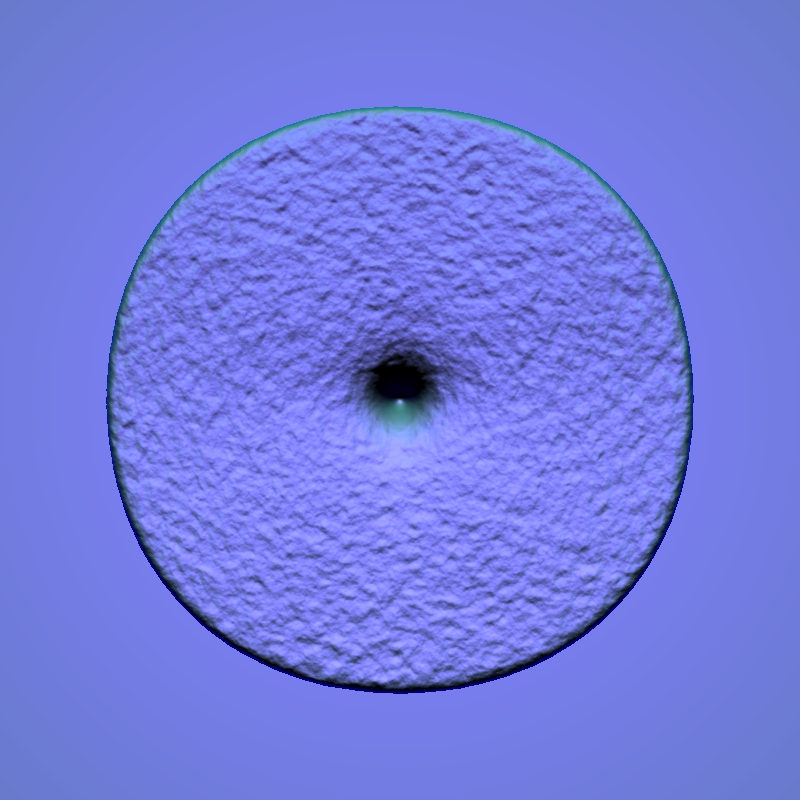}
        }\\
\subfigure[$\chi_0 = 0$, $t = 110.2$]{
            \label{set4g}
            \includegraphics[scale=0.125]{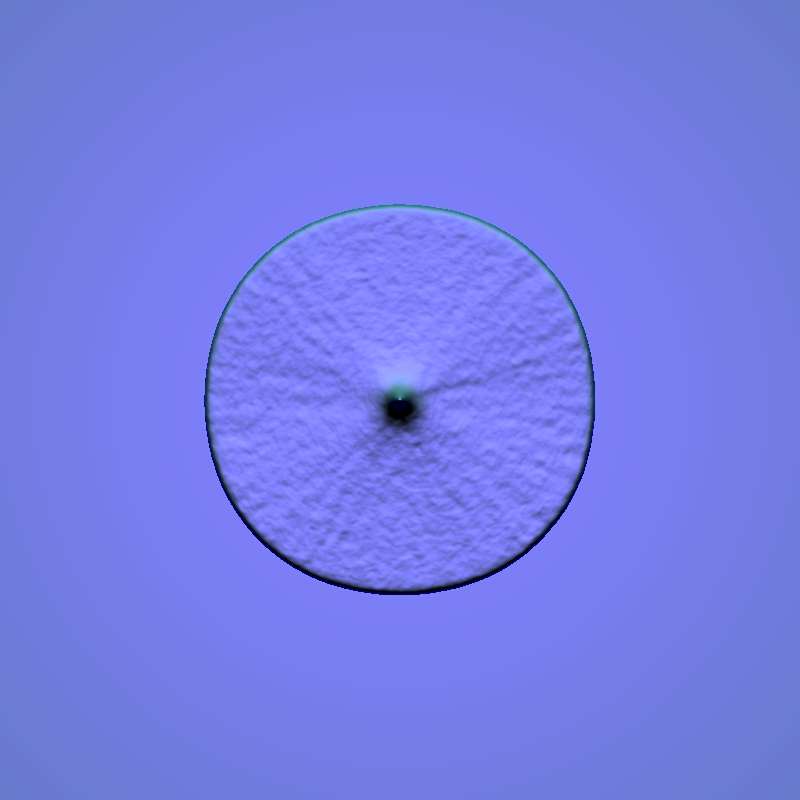}
        }
\subfigure[$\chi_0=7.5$, $t = 110.2$]{
            \label{set4h}
            \includegraphics[scale=0.125]{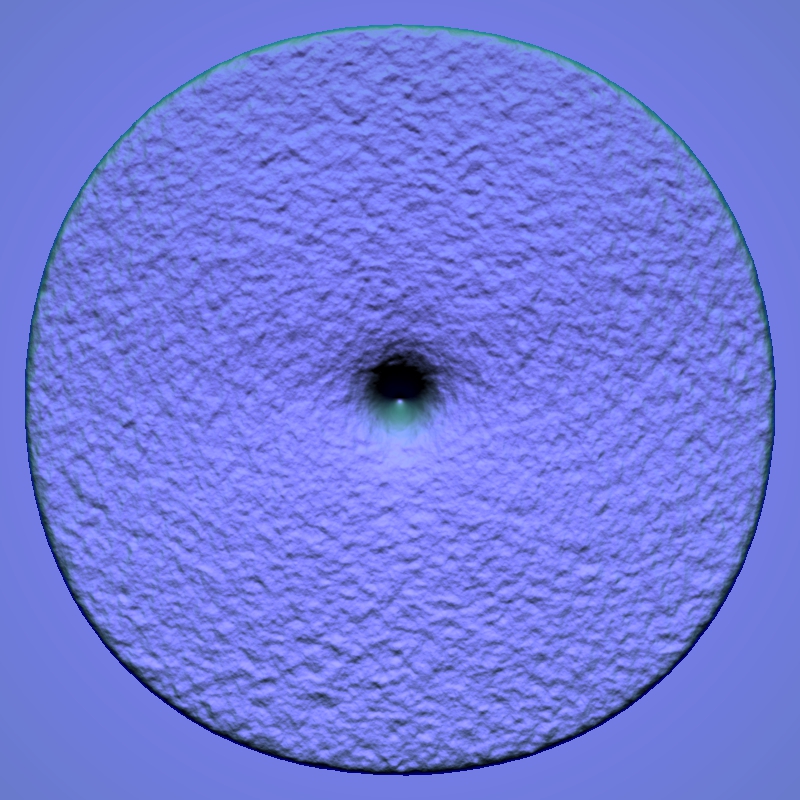}
        }\\
\end{center}
\caption{Colony growth as a result of simulations of system \eqref{ka-chem-model}, taking $\sigma_0 = 4.0, n_0 = 1.07$, with chemotaxis sensitivity $\chi_0 = 0$ (left) and $\chi_0 = 7.5$ (right), for different values of $t$.}
\label{fig:set4}
\end{figure}

\section{Conclusions}
\label{secdiscussion}

In this work we have studied the effects that nutient chemotaxis induces into bacterial population patterns modeled by a reaction-diffusion system originally proposed by Kawasaki \textit{et al.} \cite{KMMUS}. Motivated by a discussion provided by Golding \textit{et al.} \cite{GKCB}, we introduce into system \eqref{ka-model-dim} a chemotactic term compatible with the non-linear cross diffusion. After a suitable rescaling of variables, we reduce the number of free parameters of system \eqref{ka-chem-model-dim}, and we only considered $\sigma_0$, $n_0$ and $\chi_0$.  We solved numerically system \eqref{ka-chem-model}, for various values of the free parameters, including the case without chemotaxis ($\chi_0=0$). This means that we reproduced the numerical results obtained by Kawasaki \textit{et al.} in \cite{KMMUS}. Our numerical simulations of system \eqref{ka-chem-model} show that the chemotactic term provides an outward drift to the bacterial movement and therefore enhances the growth velocity of the bacterial patterns.

In order to approximate the velocity of the envelope front, under suitable assumptions, we derived a scalar equation of the bacterial density, see equation \eqref{scalar-eq-b-re}. Then, applying a result from Malaguti and Marcelli \cite{MalMar1}, see \ref{secspeedcomp}, we proved that the normal velocity $s$ is greater than the velocity of the sharp front when there is no chemotaxis ($\chi_0 =0$). Afterwards, by simulating a one-dimensional version of system \eqref{velo-system}, we calculated numerical estimations of the propagation velocity of the front.
We compared the theoretical speed predictions, defined in equations \eqref{c-bar-chi0} and \eqref{c-bar-0}, versus the estimations and we found that the propagation velocity of the envelope front is greatly increased when a chemotactic signal is present (see figure \ref{fig:Vel-1}). For the case without chemotaxis we found results comparable with those found by Kawasaki \textit{et al.} (see figure \ref{fig:Vel-xi=0}), although we must stress the fact that we are using a more accurate velocity threshold than the one used in \cite{KMMUS}. Furthermore, as figure \ref{fig:Vel-1} depicts, the numerical estimations of the propagation speed of the one-dimensional system are in good accordance with the theoretical speed thresholds that are defined for the scalar equation for $b$ (see equation \eqref{scalar-eq-b-re}), implying that the conservation-like equation \eqref{mass-conser} is a good approximation of the solutions of system \eqref{velo-system}. 

Finally, we discuss the relation of our results to the analysis of Arouh and Levine \cite{ArLe}, who studied the Kessler-Levine reaction-diffusion equations (cf. \cite{KeLe1}) and explored the addition of a nutrient chemotactic term. They found that the chemotaxis suppresses the onset of instability causing the formation of branching patterns in the colony envelope as it propagates outward. This suppression of instability is reflected in a decrease of the ``growth rate'' of the perturbation of the envelope when the chemotaxis is switched on.  The authors conjecture that this behavior can be extrapolated to other (more realistic) chemotactic models like the one introduced in this paper. It is to be pointed out, however, that our findings do not contradict the analysis of Arouh and Levine inasmuch as our numerical simulations and analytic computations exhibit an increase in the normal velocity of the envelope outward due to the chemotaxis, whereas the growth rate that Arouh and Levine refer to is the temporal eigenvalue associated to an eigenfunction (perturbation) of the linearized operator around the envelope front. Actually, a simple calculation using energy estimates (not included here) shows that in the case of our model, the spectral bounds for the eigenvalues of the linearized operator around the envelope front indeed decrease as functions of the chemotactic sensitivity. Our results and the observations of Arouh and Levine combined suggest that, although the speed of the envelope front increases with the chemotaxis (as it is shown in this paper), the growth rate of linear perturbations of the envelope is pushed to the stable (negative) region of the spectrum, delaying the formation of fingering unstable patterns (as it is shown in \cite{ArLe}). 

% If would be interesting to explore this effect of nutrient chemotaxis on the aggregation patterns from the geometric front propagation viewpoint.

%\section*{Disclosure Statement}
%The authors declare that they have no conflicts of interests with respect to their authorship and/or the publication of this article.

\appendix
%\counterwithin{prop}{section}
\section{Bounds for the velocity}
\label{secspeedcomp}

In this section we summarize the theoretical results of Malaguti and Marcelli \cite{MalMar1} that allow us to justify why the speed of the envelope front is greater when there is a chemotactic signal. Theorem 9 in \cite{MalMar1} states that for one-dimensional reaction-diffusion equations of the form
\begin{equation}\label{scalar-mm}
u_t = (D(u) u_x)_x + g(u)
\end{equation}
with $D \in C^1 ([0,1])$, $g \in C^1 ([0,1]) $, where the diffusion degenerates at $u=0$ and $u=1$, namely, 
\begin{equation}
D(0)=D(1)=0, \qquad D(u)>0  \;\;\; \text{for all} \;\; 0 < u < 1,
\end{equation}
and the rate of growth of the population is of Fisher-KPP type,
\begin{equation}
g(0) = g(1) = 0, \qquad g(u) > 0 \;\;\;\text{for all} \;\; 0 < u < 1,
\end{equation}
then there exist a traveling wave solution (unique up to translation) to equation \eqref{scalar-mm}, provided that the velocity $c$ satisfies $0< c_* \leq c $, where $c_*$ is a threshold velocity satisfying the bound 
\begin{equation}\label{threshold}
0 < c_* \leq 2 \sqrt{\sup\limits_{s \in (0,1)} \frac{D(s) g(s)}{s}}.
\end{equation}
The case $c=c_*$ corresponds to a sharp front.
% 
% 
% Malaguti and Marcelli have proved in Theorem 6 in \cite{MalMar1}, that the existence of traveling wave solutions for one-dimensional equation of the form
% \begin{equation}\label{scalar-mm}
% u_t = (D(u) u_x)_x + g(u)
% \end{equation}
% with $D \in C^1 ([0,1])$, $g \in C^1 ([0,1]) $, where the diffusion $D(u)$ satisfies,
% \begin{equation}
% D(0)=D(1)=0, \;  D(u)>0 \qquad u\in (0,1),
% \end{equation}
They prove the theorem by showing that the existence of such a traveling wave is equivalent to the existence of a solution $z=z(u)$ of the boundary-value problem
\begin{equation}
\begin{aligned}
\frac{dz}{du} &= -c - \frac{D(u) g(u)}{z},\\
z(0^+) &= z(1^-) = 0,
\end{aligned}
\end{equation}
for the the same $c$, satisfying $z(u)<0$ of all $u \in (0,1)$. 

The result can be easily extrapolated to a one-dimensional equation of the form
\begin{equation}
\label{theonedimeq}
b_t = (\tilde D(b,\chi_0)b_x)_x + g(b), 
\end{equation}
on $[0, n_0]$ and where $\tilde{D}(b,\chi_0)$ and $g(b)$ are defined in equations \eqref{D-tilde} and \eqref{reac-term}, respectively. The existence of a traveling front solution to the equation with $\chi_0 > 0$ is related to the solution to the boundary value problem
\begin{equation}
\label{problemstar}
\begin{aligned}
\frac{dz}{db} &= -c - \frac{\tilde D(b,\chi_0) g(b)}{z},\\
z(0^+) &= z(n_0^-) = 0.
\end{aligned}
\end{equation}
Likewise the traveling wave without chemotaxis is linked to the solutions to the problem
\begin{equation}
\label{problemzero}
\begin{aligned}
\frac{dz}{db} &= -c - \frac{\tilde D(b,0) g(b)}{z},\\
z(0^+) &= z(n_0^-) = 0.
\end{aligned} 
\end{equation}

In order to compare the speed of the envelope front in both cases, we consider the following sets
\begin{equation}
\label{defAchi0}
A_{\chi_0} = \{ c > 0 \, : \, \text{problem \eqref{problemstar} has a non-positive solution}\},
\end{equation}
\begin{equation}
\label{defA0}
A_{0} = \{ c > 0 \, : \, \text{problem \eqref{problemzero} has a non-positive solution}\},
\end{equation}
that define the range of values of $c$ for which there exists a traveling wave solution for equation \eqref{scalar-eq-b-re} (see Theorem 9 in \cite{MalMar1}). It should be noted that the infimum of the set, in either case, corresponds to the threshold speed $c_*$. The following proposition compare the two velocity thresholds (with and without chemotaxis).

\begin{prop}\label{propbound}
For any fixed value of $n_0 > 0$ and for all values of $\chi_0 \geq 0$ there hold,
\begin{itemize}
 \item[(i)] $\bar c(\chi_0) \geq \bar c(0)$, and
 \item[(ii)] $c_*(\chi_0) \geq c_*(0)$.
\end{itemize}
\end{prop}
\begin{proof}
The proof of (i) is straightforward: for all values of $\chi_0 \geq 0$ and for each $b \in (0,n_0)$ there holds
\[
 \psi(b,\chi_0) = b(1-b/n_0)^2(1+\chi_0 b) \geq \psi(b,0) = b(1-b/n_0)^2.
\]
Therefore $\max_{b\in[0,n_0]} \psi(b,\chi_0) \geq \max_{b\in[0,n_0]} \psi(b,0)$, yielding (i).

To show (ii), let us suppose that $c > 0$ is a velocity such that $c \in A_{\chi_0}$ for some $\chi_0 > 0$. This means that the problem \eqref{problemstar} has a non-positive solution $\zeta(b)$. Therefore, since $\tilde D(b,\chi_0) > \tilde D(b,0)$ for all $b \in (0,n_0)$ inasmuch as $\chi_0 > 0$, we have that
\[
 \frac{d\zeta}{db} = -c - \frac{\tilde D(b,\chi_0)}{\zeta(b)} > -c - \frac{\tilde D(b,0)}{\zeta(b)}.
\]
By Lemma 8 in \cite{MalMar1}, there exists a negative solution $z = z(b)$ to problem \eqref{problemzero}, with the same constant $c$. This shows that $c \in A_0$. Therefore, $A_{\chi_0} \subseteq A_0$ for all $\chi_0 \geq 0$, and consequently $c_*(\chi_0) = \inf A_{\chi_0} \geq c_*(0) = \inf A_0$. This yields (ii).
\end{proof}

\section*{Acknowledgements}

The work of J. F. Leyva was partially supported by CONACyT (Mexico) through a scholarship for doctoral studies, grant no. 220865.

%\bibliographystyle{elsarticle-num}
%
%\bibliography{references_m}
%\bibliographystyle{unsrt}

% \bibliographystyle{amsplain}
% \bibliographystyle{siam}
% \bibliographystyle{amsalpha}
% \bibliographystyle{newstyle}
% \bibliographystyle{unsrt}

\end{document}